\documentclass[10pt,journal,a4paper]{IEEEtran} %TWC

\usepackage{amssymb}
\usepackage{amsmath}
\usepackage{amsmath,bm}
\usepackage{amsthm}
\usepackage{graphicx}
\usepackage{cuted}
\usepackage{subfigure}
\usepackage{multirow}
\usepackage{makecell} 
\usepackage{cite}
\usepackage{enumerate}
\usepackage{enumitem}
\usepackage{stfloats}
\usepackage{algorithm}
\usepackage{algorithmic}  
\usepackage{color, soul}
\usepackage{hyperref}
\usepackage[labelsep=period]{caption}


\allowdisplaybreaks[4]
\DeclareMathOperator{\Tr}{Tr}
\DeclareMathOperator{\rank}{Rank}
\DeclareMathOperator{\newvec}{vec}
\DeclareMathOperator{\diag}{diag}

\begin{document}
	\newtheorem{theorem}{\bf~~Theorem}
	\newtheorem{remark}{\bf~~Remark}
	\newtheorem{property}{\bf~~Property}
	\newtheorem{observation}{\bf~~Observation}
	\newtheorem{definition}{\bf~~Definition}
	\newtheorem{lemma}{\bf~~Lemma}
	\newtheorem{preliminary}{\bf~~Preliminary}
	\newtheorem{proposition}{\bf~~Proposition}
	\newtheorem{comment}{\bf~~Comment}
	\renewcommand\arraystretch{0.9}
%	\title{Reconfigurable Refractive Surface-Enabled Multi-User Holographic MIMO Communications}
	\title{\Huge{Holographic Beamforming for Integrated Sensing and Communication with Mutual Coupling Effects}}
%	 \title{\setlength{\baselineskip}{3pt}\Large Title}
%	 %\vspace{0.2cm}
%\title{{ Intelligent Omni-Surfaces: Reflection-Refraction Model, Full-dimensional Beamforming, and System Implementation}}
%\title{{ Intelligent Omni-Surfaces: Reflection-Refraction Model, Full-dimensional Beamforming, and System Implementation}}
	\author{\normalsize \IEEEauthorblockN{
			{Shuhao Zeng}, \IEEEmembership{\normalsize  Member, IEEE},
			{Haobo Zhang}, \IEEEmembership{\normalsize Member, IEEE},
			{Boya Di}, \IEEEmembership{\normalsize Member, IEEE},
			{Hongliang Zhang}, \IEEEmembership{\normalsize Member, IEEE},\\
			{Zijian Shao}, \IEEEmembership{\normalsize Member, IEEE},
			{Zhu Han}, \IEEEmembership{\normalsize Fellow, IEEE},
			{H. Vincent Poor}, \IEEEmembership{\normalsize Life Fellow, IEEE}, and
			{Lingyang Song}, \IEEEmembership{\normalsize Fellow, IEEE}}
		\thanks{Shuhao Zeng is with Department of Electrical and Computer Engineering, Princeton University, NJ, USA, and also with School of Electronic and Computer Engineering, Peking University Shenzhen Graduate School, Shenzhen, China (email: shuhao.zeng96@gmail.com).}		
		\thanks{Haobo Zhang is with School of Electronic and Computer Engineering, Peking University Shenzhen Graduate School, Shenzhen, China, and also with Department of Engineering, University of Cambridge, Cambridge, UK (email: haobo.zh97@gmail.com)}
		\thanks{Boya Di, and Hongliang Zhang are with School of Electronics, Peking University, Beijing 100871, China (email: boya.di@pku.edu.cn; hongliang.zhang@pku.edu.cn; lingyang.song@pku.edu.cn).}
		\thanks{Zijian Shao, and H. Vincent Poor are with Department of Electrical and Computer Engineering, Princeton University, NJ, USA (email: shuhao.zeng96@gmail.com; zijianshao@princeton.edu; poor@princeton.edu).}
		\thanks{Zhu Han is with Electrical and Computer Engineering Department, University of Houston, Houston, TX, USA, and also with the Department of Computer Science and Engineering, Kyung Hee University, Seoul, South Korea (email: zhan2@uh.edu).}
		\thanks{Lingyang Song is with School of Electronic and Computer Engineering, Peking University Shenzhen Graduate School, Shenzhen, China, and also with School of Electronics, Peking University, Beijing 100871, China, and also with Hunan Institute of Advanced Sensing and Information Technology, Xiangtan University, Xiangtan, China (email: lingyang.song@pku.edu.cn).}		
	}
	
	\maketitle
\begin{abstract}
	Integrated sensing and communication~(ISAC) is envisioned as a key technology in 6G networks, owing to its potential for high spectral and cost efficiency. As a promising solution for extremely large-scale arrays, reconfigurable holographic surfaces~(RHS) can be integrated with ISAC to form the holographic ISAC paradigm, where enlarged radiation apertures of RHS can achieve significant beamforming gains, thereby improving both communication and sensing performance. In this paper, we investigate holographic beamforming designs for ISAC systems. However, unlike existing holographic beamforming schemes developed for RHS-aided communications, such designs require explicit consideration of mutual coupling effects within RHS. This is because, unlike prior works only considering communication performance, ISAC systems incorporate sensing functionality, which is sensitive to sidelobe levels. Ignoring mutual coupling in holographic beamforming can lead to notable undesired sidelobes, thus degrading sensing performance. The consideration of mutual coupling introduces new challenges, i.e., it induces non-linearity in beamforming problems, rendering them inherently non-convex. To address this issue, we propose a tractable electromagnetic-compliant holographic ISAC model that characterizes mutual coupling in a closed form using coupled dipole approximations. We then develop an efficient holographic beamforming algorithm to suppress sidelobes and enhance ISAC performance. Numerical results validate effectiveness of the proposed algorithm.
\end{abstract}

	\begin{IEEEkeywords}
		Holographic ISAC, holographic beamforming, mutual coupling effects, coupled dipole approximations
	\end{IEEEkeywords}

\section{Introduction}
	
Integrated sensing and communications~(ISAC) is expected to play a key role in the future sixth generation~(6G) networks~\cite{Zhang_overview_2021}. Unlike conventional separate communication and radar systems operating in isolated frequency bands~\cite{Moghaddasi_Multifunctional_2016}, communication and sensing functionalities share the same radio resources and hardware infrastructure in ISAC systems~\cite{Hassanien_dual_2019,Liu_Survey_2022}. Therefore, a significant enhancement of spectral, energy, and cost efficiency can be achieved by ISAC systems as compared to conventional approaches~\cite{Cui_integrating_2021,Feng_Joint_2020}.

To further improve communication and sensing performance, recent studies have explored the application of reconfigurable holographic surfaces~(RHS) to ISAC systems, giving rise to a new paradigm known as holographic ISAC~\cite{Liu_HISAC}. Specifically, an RHS is a representative type of metasurface, consisting of densely packed sub-wavelength metamaterial elements embedded in a waveguide structure excited by multiple feeds~\cite{Deng_VTM_2023,Zeng_dual_2024}. Electromagnetic~(EM) waves generated by the feeds propagate along the waveguide and sequentially excite the RHS elements. Each RHS element controls the amplitude of the radiated wave by adjusting the bias voltages applied to its onboard diode, thereby forming directional beams. This beamforming technique is referred to as holographic beamforming. Unlike conventional phased arrays that rely on costly hardware, RHS enables low-cost and energy-efficient beamforming based on simple diodes~\cite{Zhang_Target_2024,Pizzo_HMIMO_2020}. This makes RHS particularly suitable for large-scale array implementations. The resulting large radiation aperture of the RHS enables high directive gain~\cite{Ji_EM_2024,Zeng_RRS_2024}, thereby enhancing both communication throughput and sensing accuracy in ISAC systems.

Initial works have demonstrated the effectiveness and practicality of the holographic ISAC paradigm~\cite{Zhao_performance_2025,Gavras_duplex_2023,Hu_multi_band_2023,Zhang_HISAC_2022}. In~\cite{Zhao_performance_2025}, the authors investigate a holographic ISAC system, where closed-form expressions are derived for the sensing rate, communication rate, and outage probability under both instantaneous and statistical channel state information. Numerical results demonstrate that the holographic ISAC system outperforms conventional MIMO-based ISAC systems in terms of both sensing and communication performance. Then, the authors in~\cite{Gavras_duplex_2023} and~\cite{Hu_multi_band_2023} further extend the scenario to full-duplex and wideband holographic ISAC systems, respectively, and demonstrate the capability of the holographic ISAC systems to combat in-band interference in full-duplex operation and the high frequency selectivity of metamaterial elements in wideband systems. The practicality of such methodology is further experimentally verified in~\cite{Zhang_HISAC_2022}, where an RHS-based holographic ISAC hardware prototype is built and its positioning range and communication rate are measured.

In this paper, we aim to investigate the holographic beamforming design for ISAC systems. However, different from existing holographic beamforming schemes for RHS-aided communication systems, the holographic beamforming design for ISAC system should take the mutual coupling effect within the RHS into account, due to the integration of sensing functionality into the communication network. Specifically, within the RHS, the electromagnetic~(EM) fields radiated by one RHS element couple to adjacent elements, inducing additional currents within them. These time-varying currents, in turn, generate scattered fields that impinge upon the original element and excite new currents. Such phenomenon is known as \emph{mutual coupling}. The ignorance of the mutual coupling effects in holographic beamforming can bring notable undesired sidelobes, as demonstrated in Section~\ref{subsection_impact}. Existing works on holographic beamforming for RHS-aided communication networks aim to improve data rate, which is mainly determined by the main lobe gain rather than the sidelobe level\footnote{Although the sidelobe of the RHS deployed at the BS can lead to inter-user interference in multi-user scenarios, such interference can be effectively mitigated by directing the sidelobe away from the other users through digital beamforming.}, and thus they can neglect the mutual coupling effect for simplicity without compromising the communication performance. Differently, the holographic ISAC system integrates additional sensing functions. Note that power leakage from sidelobes can propagate to various environmental clutterers, such as trees and buildings, where the resulting echo signals can cause interference to sensing procedures. Therefore, the radar sensing performance of holographic ISAC systems is sensitive to sidelobe levels\footnote{Due to the potentially large number and wide spatial distribution of environmental clutterers, it is challenging to align the nulls of the RHS-generated beam with all clutterer directions through digital beamforming. Therefore, it becomes essential to suppress sidelobe levels by explicitly accounting for mutual coupling effects in the design of holographic beamforming.}, which thus necessitates the consideration of mutual coupling effects in holographic beamforming designs.

%Since there can be numerous and widespread environmental clutters, it is hard to align the nulls of the generated beam of the RHS with all environmental clutters through digital beamforming to suppress undesired echo signals. Therefore, it is necessary to suppress sidelobe levels by considering the mutual coupling effects in holographic beamforming designs.

However, the consideration of the mutual coupling effects in the holographic beamforming for ISAC systems introduces the following new challenges:
\begin{itemize}
	\item First, the development of a tractable mutual coupling aware holographic ISAC model is non-trivial, since it is difficult to characterize the mutual coupling effect within the RHS in closed form.
	\item Second, mutual coupling effects induce non-linearity in the holographic beamforming problem\footnote{As we can see from (\ref{Holographic_BF}) and (\ref{tx_signal_RHS}), when the mutual coupling effect within the RHS is considered, the transmitted communication and sensing signals of the RHS are non-linear with respect to holographic patterns at the RHS, which is different from existing holographic ISAC systems aided by an ideal RHS without the mutual coupling effect~\cite{Deng_RHS_multi_user_2022}.}, leading to inherent non-convexity and making it challenging to solve.
\end{itemize}
To tackle these challenges, we consider a holographic ISAC system where a BS employs an RHS to realize holographic beamforming, enabling communication with multiple users while simultaneously sensing targets in various directions of interest. To facilitate holographic beamforming design, we propose an EM-compliant mutual coupling aware RHS-aided holographic ISAC model using coupled dipole approximations. Based on this model, we formulate a holographic beamforming problem to suppress sidelobes. An efficient mutual coupling aware algorithm is proposed to solve this non-convex problem. Simulation results demonstrate the effectiveness of the proposed algorithm. The main contributions of this paper are summarized below. 
\begin{itemize}
	\item We consider a holographic ISAC system, where a BS equipped with an RHS serves multiple communication users and detects multiple radar targets. Based on coupled dipole approximations, a tractable EM-compliant RHS-aided ISAC model is proposed, which explicitly accounts for mutual coupling effects in closed form. We demonstrate the generalizability and effectiveness of the proposed model. Building on this model, we further explain the mechanism through which mutual coupling effects degrade sensing performance of ISAC systems.

	\item Based on the model, we formulate a holographic beamforming problem to improve ISAC performance by suppressing sidelobes. The problem is decomposed into two subproblems, which are solved iteratively using a joint mutual coupling aware optimization algorithm. Specifically, the digital beamforming subproblem is reformulated as a more tractable quadratic constrained quadratic program~(QCQP) via fractional programming. For the holographic beamforming subproblem, the Neumann series is employed to convert it into a sequence of QCQP formulations.
	
	\item Simulation results demonstrate the efficacy of the proposed algorithm. The effects of transmit power, the number of users, and the number of sensing targets are also analyzed. Furthermore, the performance of an ISAC system assisted by an RHS, while considering the impact of mutual coupling, is compared with that of a conventional phased array under the same hardware cost.
\end{itemize}

The rest of the paper is organized as follows. In Section~\ref{sec_mutual_coupling}, we characterize the mutual coupling effect within the RHS using coupled dipole approximations. Building on this, we introduce a mutual coupling aware RHS-aided ISAC model in Section~\ref{sec_system_model}. The impact of mutual coupling on holographic beamforming is then analyzed in Section~\ref{subsection_impact}. Section~\ref{sec_problem_formulation} formulates a max-min rate optimization problem, which is decomposed into two subproblems. To address these subproblems, we propose a mutual coupling aware beamforming algorithm in Section~\ref{sec_optimization_algorithm}. Simulation results are presented in Section~\ref{sec_simulation}, followed by concluding remarks in Section~\ref{sec_conclusion}. 
	
	\section{Characterizing Mutual Coupling Via Coupled Dipole Approximations}
	\label{sec_mutual_coupling}
%	\subsection{Mutual Coupling in RHS}
	\begin{figure}[!t]
		\centering
		\includegraphics[width=0.47\textwidth]{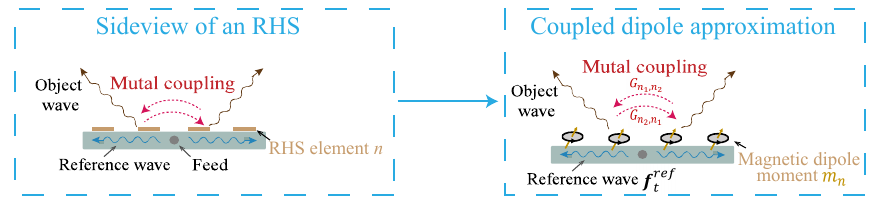}
		%		\vspace{1.5mm}
		\caption{An illustration of coupled dipole approximations for capturing the mutual coupling effects within an RHS array.}
				\vspace{-3mm}
		\label{illustration_CDM}
	\end{figure}
	
%	Introduce the concept of coupling. 
%	
%	Further, we demonstrate that the mutual coupling effect is significant in RHS through simulations. Specifically, given a fixed configuration of RHS elements, we obtain the beam pattern both with and without considering the mutual coupling effect, and then compare these patterns with the ground truth to highlight the significance of mutual coupling effects.
%	\subsection{Coupled Dipole Approximations}
	\label{dipole_approximation}
%	To characterize the mutual coupling within the RHS, the coupled dipole model is an effective model, which can be described by the following expression,
%	\begin{align}
%		\bm{G}\bm{m}_l=\bm{f}^{inc}_l,
%	\end{align}
%	where $\bm{f}^{inc}_l$ is incident magnetic fields when a signal is fed into the $l$-th feed of the RHS, while $\bm{m}_l \in \mathbb{C}^{N\times 1}$ represents corresponding effective magnetic dipole moments. The off-diagonal entries of $\bm{G}\in\mathbb{C}^{N\times N}$ are the Green's functions and the diagonal entries are the inverse of the polarizabilities, i.e., $(\alpha_n)^{-1}$. When the biased voltages applied to the RHS elements change, the magnitude of the polarizability corresponding to this element also becomes different, i.e., the radiation amplitudes of the RHS elements vary. By combining incident magnetic fields and dipole moments corresponding to different feeds, we have the following expression,
%	\begin{align}
%		\label{coupled_dipole_model}
%		\bm{G}\bm{M}=\bm{F}^{inc},
%	\end{align}
%	where $\bm{M}=[\bm{m}_1,\dots,\bm{m}_L]$ and $\bm{F}^{inc}=[\bm{f}^{inc}_1,\dots,\bm{f}^{inc}_L]$. Based on (\ref{coupled_dipole_model}), the dipole
%	$\bm{G}$
%	To characterize the mutual coupling effect present in the reconfigurable holographic surface~(RHS), the coupled dipole model is an effective analytical framework~\cite{Mancera_Polarizability_2017}. 
	{As shown in Fig.~\ref{illustration_CDM}, the core idea is to represent each RHS element as a magnetic dipole given that RHS elements are electrically small~\cite{Mancera_Polarizability_2017,Yue_2024}}. Further, the equivalent dipoles are assumed to interact with each other through coupled fields, so as to form a coupled dipole system~\cite{Mancera_analytical}. Since it is easier to analyze the influence of mutual coupling on the overall radiation characteristics for a coupled dipole system, such an approach provides a more convenient and mathematically tractable means to characterize mutual coupling effects in the RHS.

%	while it is technically easier to characterize the influence of mutual coupling on the overall radiations for a coupled dipole system.
	
%	it is easier to characterize the influence of mutual coupling 
	
	Specifically, let $N$ denote the number of RHS elements, and thus the number of equivalent magnetic dipoles. As established in~\cite{Yoo_Sub_2023}, such a coupled dipole system is governed by the following matrix equation, i.e.,
	\begin{align}
		\label{coupled_dipole_model}
		%		\bm{G}\bm{m}=\bm{f}^{inc},
		\left((e^{j\tau}\bm{\Theta})^{-1}-\bm{G}\right)\bm{m}=\bm{f}^{ref}_t,
	\end{align}
%	where $\bm{f}^{ref}_l=[f_{n,l}^{ref}]_n\in \mathbb{C}^{N\times 1}$ records the distribution of reference waves generated by the $l$-the feed over the $N$ equivalent dipoles, $\bm{m}=[m_n]_n \in \mathbb{C}^{N\times 1}$ represents the vector of magnetic dipole moments for all equivalent dipoles, $\bm{\Theta}=\diag\{\theta_1,\dots,\theta_N\}$ is a diagonal matrix containing the polarizability of all equivalent dipoles, and $\bm{G}\in\mathbb{C}^{N\times N}$ is coupling intensity matrix. The detailed derivation of the model in (\ref{coupled_dipole_model}) can be found in Appendix~\ref{app_derivation_CDM}. In the following, we will elaborate on these key components of the coupled dipole model and their connections to holographic beamforming, i.e.,
	the physical meanings of $(e^{j\tau}\bm{\Theta})$, $\bm{G}$, $\bm{m}$, and $\bm{f}^{ref}_t$, as well as their connections with holographic beamforming are introduced as follows. The detailed derivation process of (\ref{coupled_dipole_model}) can be found in~\cite{Mancera_analytical}.

	\begin{itemize}
	\item Reference wave distribution $\bm{f}^{ref}_t=[f_{n,t}^{ref}]_n\in \mathbb{C}^{N\times 1}$: $\bm{f}^{ref}_t$  represents the distribution of reference waves generated by the $t$-th feed across the $N$ equivalent dipoles. Specifically, the EM wave generated by the $t$-th feed propagates along the waveguide and sequentially excites the equivalent dipoles, which is referred to as reference wave. 
	
	\item Magnetic dipole moments $\bm{m}=[m_n]_n\in \mathbb{C}^{N\times 1}$: $\bm{m}$ contains the magnetic moments of all equivalent magnetic dipoles, which determine the \emph{object wave} generated by the RHS. Specifically, as a key parameter of the equivalent magnetic dipoles, the magnetic moment $m_n$ governs the EM wave radiated by the $n$-th magnetic dipole (and consequently that by the $n$-th RHS element) and exhibits a linear relationship with it. Further, the EM waves generated by different RHS elements propagate into free space and superimpose to form the {object wave}.
	
	\item Polarizability matrix $e^{j\tau}\bm{\Theta}=\diag\{\theta_1e^{j\tau},\dots,\theta_Ne^{j\tau}\}\in \mathbb{C}^{N\times N}$: The polarizability $\theta_ne^{j\tau}$ of the $n$-th equivalent magnetic dipole is defined as the sensitivity of the magnetic moment $m_n$ to the local field $H_n$, given by
	\begin{align}
		\label{def_polarizability}
		\theta_ne^{j\tau}=m_n/H_n,
	\end{align}
	where $\theta_n$ and $\tau$ represent the magnitude and phase of polarizability of the $n$-th magnetic dipole, respectively. The magnitudes of polarizability for the equivalent dipoles, denoted as $\bm{\Theta}=\{\theta_1,\dots,\theta_N\}\in \mathbb{R}^{N\times N}$, can be interpreted as the \emph{holographic pattern} of the RHS in the context of mutual coupling. This is because the magnitude $\theta_n$ of polarizability is reconfigurable and depends on the biased voltages applied to the $n$-th RHS element~\cite{Boyarsky_Ele_2021}. Further, according to (\ref{coupled_dipole_model}), $\bm{\Theta}$ influences the transmitted signals of the RHS by modulating the magnetic moments $\bm{m}$. Such interpretation is further justified in Remark~\ref{remark_eff}. 
	
	\item Coupling matrix $\bm{G}\in\mathbb{C}^{N\times N}$: $\bm{G}$ records the coupling coefficients among different RHS elements, where its off-diagonal entries contain the Green's function $G_{n',n}$ while the diagonal elements are all zero. Specifically, the EM waves radiated by the $n$-th equivalent dipole propagate to another dipole $n'$, generating coupled fields $H_{n\rightarrow n'}^{MC}$. These coupled fields, along with the reference wave, excite the $n'$-th dipole. Here, the coupled field $H_{n\rightarrow n'}^{MC}$ is influenced by both the propagation environment and the source dipole $n$, with the effect of the propagation environment captured by the Green's function $G_{n,n'}$. {As the RHS element spacing increases, the amplitudes of $G_{n',n}$ ($n\neq n'$) tend to decrease. This is because the EM wave generated by the $n$-th RHS element can propagate to the $n'$-th RHS element through both free space outside of the RHS and the waveguide within the RHS~\cite{Mancera_analytical}. Note that the amplitude of the EM wave tend to decrease for both propagations through the free space and within the waveguide~\cite{Balanis_antenna_the}. Therefore, $|G_{n',n}|$ tends to decrease as the RHS element spacing increases.}
\end{itemize} 

By mapping the coupled dipole model in (\ref{coupled_dipole_model}) onto the conventional RHS-based holographic beamforming model~\cite{Deng_RHS_multi_user_2022}, we can reinterpret the holographic beamforming process within the framework of coupled dipole approximations, as summarized in the following remark.
	
	\begin{remark}
		\label{interpretation}
		(Reinterpretations of holographic beamforming in the context of coupled dipole approximations) The $t$-th feed generates reference wave $\bm{f}^{ref}_t$, which acts as the initial excitation for all magnetic dipoles. By taking into account both the holographic pattern $\bm{\Theta}$ and mutual coupling characterized by $\bm{G}$, we can obtain the magnetic moments $\bm{m}$ of all equivalent magnetic dipoles. The magnetic moments $\bm{m}$ determine the radiated fields of the overall dipole system, where the generated waves of different magnetic dipoles propagate into free space and superimpose to form the overall object wave. 
%		 which are proportional to the transmitted signals of the RHS elements, i.e., the object wave.
	\end{remark}
	
	Based on the above coupled dipole model, we propose an EM-compliant mutual coupling aware RHS-aided ISAC model to facilitate further holographic beamforming design, as shown in the next section.
%	To facilitate further holographic beamforming design, we integrate the coupled dipole model in (\ref{coupled_dipole_model}) into RHS-aided communication models by deriving the holographic beamforming matrix, as shown in the next section.

	\section{System Model}
	\label{sec_system_model}
	In this section, we first introduce the RHS-aided ISAC system in clutter environments. Then, we propose a mutual coupling aware holographic beamforming model and present the data transmission scheme. Afterwards, we discuss the model's effectiveness in capturing the physical behaviour of mutual coupling effects within the RHS.
	\subsection{Scenario Description}
	\begin{figure}[!t]
		\centering
		\includegraphics[width=0.47\textwidth]{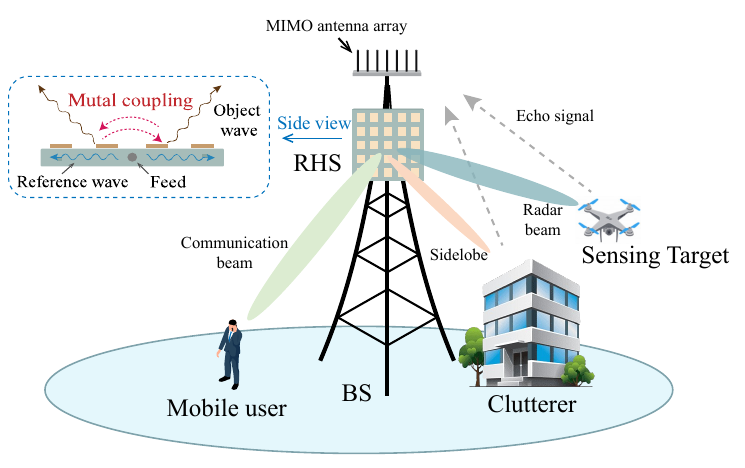}
		%		\vspace{1.5mm}
		\caption{System model of an RHS-aided ISAC network with mutual coupling effects.}
%		\vspace{-3mm}
		\label{sysmodel}
	\end{figure}
	
	As shown in Fig.~\ref{sysmodel}, we consider a holographic ISAC system that consists of $L$ mobile users denoted by $\mathbb{L}= \{1,\cdots,l,\cdots,L\}$, multiple targets, and a BS equipped with an RHS and a MIMO antenna array. Further, we assume there are $W$ clutterers in environments, where leakage power from the sidelobes of the RHS can propagate to various clutterers and the resulting echo signals can cause interference to sensing procedures. 
	
	To enable both sensing and communication functionalities, the system operates in three sequential steps~\cite{Haobo_HISAC_2022}. First, the BS optimizes the ISAC signals to maximize the minimum data transmission rate among the mobile users while ensuring the target radar performance. Next, the RHS transmits the ISAC signals with multiple beams directed toward the users and targets. Finally, the communication users receive and decode the transmitted signals to extract communication information, while the MIMO antenna array captures the echo signals reflected by the targets for radar sensing.
%	\begin{itemize}
%		\item \textit{Optimization}: The BS optimizes the ISAC signals to maximize the radar performance and to assure the target quality of service~(QoS) to communication users.
%		\item \textit{Transmission}: The RHS transmits the ISAC signals with multiple beams towards the directions of the users and targets.
%		\item \textit{Reception}: The communication users receive the signals from the RHS and decode the received signals to obtain the communication information. At the same time, the MIMO antenna array listens to the echo signals reflected by the targets for radar sensing
%	\end{itemize}
	
%	In traditional MIMO systems~[13], multiple beams are emitted by large antenna arrays whose elements are connected to many RF chains or phase shifters, resulting in complex structures that require costly circuits and hardware. In holographic beamforming based ISAC systems, by contrast, the beams with desired properties are generated by an RHS that is made of closely spaced radiating elements and a limited number of RF chains, which provide powerful beam-steering capabilities~[34] and significantly reduce the complexity and cost~[24].

		\subsection{Mutual Coupling Aware Holographic Beamforming Model}	
	Based on the coupled dipole model in (\ref{coupled_dipole_model}), we will try to construct a mutual coupling aware holographic beamforming matrix $\bm{B}\in\mathbb{C}^{N\times T}$ in this part. Here, $t$-th column of the holoraphic beamformer $\bm{B}$, i.e., $\bm{b}_t$, represents the contribution of the input signal of the $t$-th feed on the radiated signals of the RHS elements. 
	
%	Specifically, define $\xi_n$ and $\tau$ as the magnitude and phase of the polarizability $\theta_n$ of the $n$-th magnetic dipole, i.e., $\xi_n=\theta_n$ and $\tau=\angle\theta_n$. By stacking the polarizability magnitude $\xi_n$ of different dipoles as a diagonal matrix $\bm{\Xi}$, i.e., $\bm{\Xi}=\diag\{\xi_1,\dots,\xi_N\}$, we have
%	\begin{align}
%		\label{theta_expansion}
%			\bm{\Theta}=e^{j\tau}\bm{\Xi}.
%%			\exp(j\tau)\bm{\Xi}.
%		\end{align}
%	Here, the matrix $\bm{\Xi}$ represents holographic pattern according to the discussion in Section~\ref{sec_mutual_coupling}. 
	
	Based on (\ref{coupled_dipole_model}) and Remark~\ref{interpretation}, the holographic beamformer $\bm{b}_t$ corresponding to the $t$-th feed can be modeled by
	\begin{align}
		\label{holographic_BF_single_feed}
		\bm{b}_t=k\left((e^{j\tau}\bm{\Theta})^{-1}-\bm{G}\right)^{-1} \widetilde{\bm{f}}_{t}^{ref}.
	\end{align}
	Here, $k$ is a complex constant relating the magnetic moments $\bm{m}$ of the equivalent dipoles to the transmitted signals $\bm{x}$ of the RHS elements, i.e., $\bm{x}=k\bm{m}$, which is independent of the configurations of the RHS elements. The derivation of $k$ is provided in Appendix~\ref{app_eff}.
	
	Then, by stacking the vector $\bm{b}_t$, the holographic beamforming matrix can be given by
	\begin{align}
		\label{Holographic_BF}
		\bm{B}=k\left((e^{j\tau}\bm{\Theta})^{-1}-\bm{G}\right)^{-1}\widetilde{\bm{F}}^{ref},
	\end{align}
	where $\widetilde{\bm{F}}^{ref}=[\widetilde{\bm{f}}_{1}^{ref},\dots,\widetilde{\bm{f}}_{T}^{ref}]$, with $T$ representing the number of feeds of the RHS.
	
	To investigate the influence of mutual coupling effects on the communication model, in the following remark, we compare the proposed mutual coupling aware holographic beamformer in (\ref{Holographic_BF}) against an existing one developed for an ideal RHS without mutual coupling effects~\cite{Deng_RHS_multi_user_2022}. 
	\begin{remark}\label{remark_compare_ideal}
		(Comparison with ideal holographic beamformer) For an ideal RHS without mutual coupling effects, the holographic beamformer can be written as~\cite{Deng_RHS_multi_user_2022},
		\begin{align}
			\label{ideal_Holographic_BF}
			\bm{B}^{ideal}=\bm{M}\widetilde{\bm{F}}^{ref},
		\end{align}
		where $\bm{M}$ is a diagonal matrix, with the diagonal elements being radiation amplitudes of the RHS elements. 
		
		By comparing (\ref{Holographic_BF}) and (\ref{ideal_Holographic_BF}), we can see that the presence of mutual coupling effects brings additional coupling matrix $\bm{G}$. This matrix reflects how the transmitted signals from the RHS elements is coupled to other RHS elements.
		%		, thereby inducing additional currents within them.
		%		Such  which indicates that the transmitted signals of the RHS elements can couple to other RHS elements and excite addition currents within them. 
		
		%		The additional matrix $\bm{G}$ and corresponding inverse operation involved in the holographic beamformer in (\ref{Holographic_BF}) induced by mutual coupling effects bring non-linearity, which further complicates the optimization of holographic beamformer.
	\end{remark}
	Then, we further show the generalizability of the proposed model, as shown in Property~\ref{remark_eff}.
	\begin{property}
		\label{remark_eff}
		(Generalizability of proposed holographic beamformer) Although the proposed holographic beamformer in (\ref{Holographic_BF}) is intended for an RHS with mutual coupling effects, it can also be applied to an ideal case without mutual coupling, which, thus demonstrates its generalizability.
	\end{property}
	\begin{proof}
		See Appendix~\ref{app_eff}.
	\end{proof}

	\subsection{Integrated Sensing and Communication Scheme}
	\begin{figure*}[!t]
		\centering
		\includegraphics[width=0.7\textwidth]{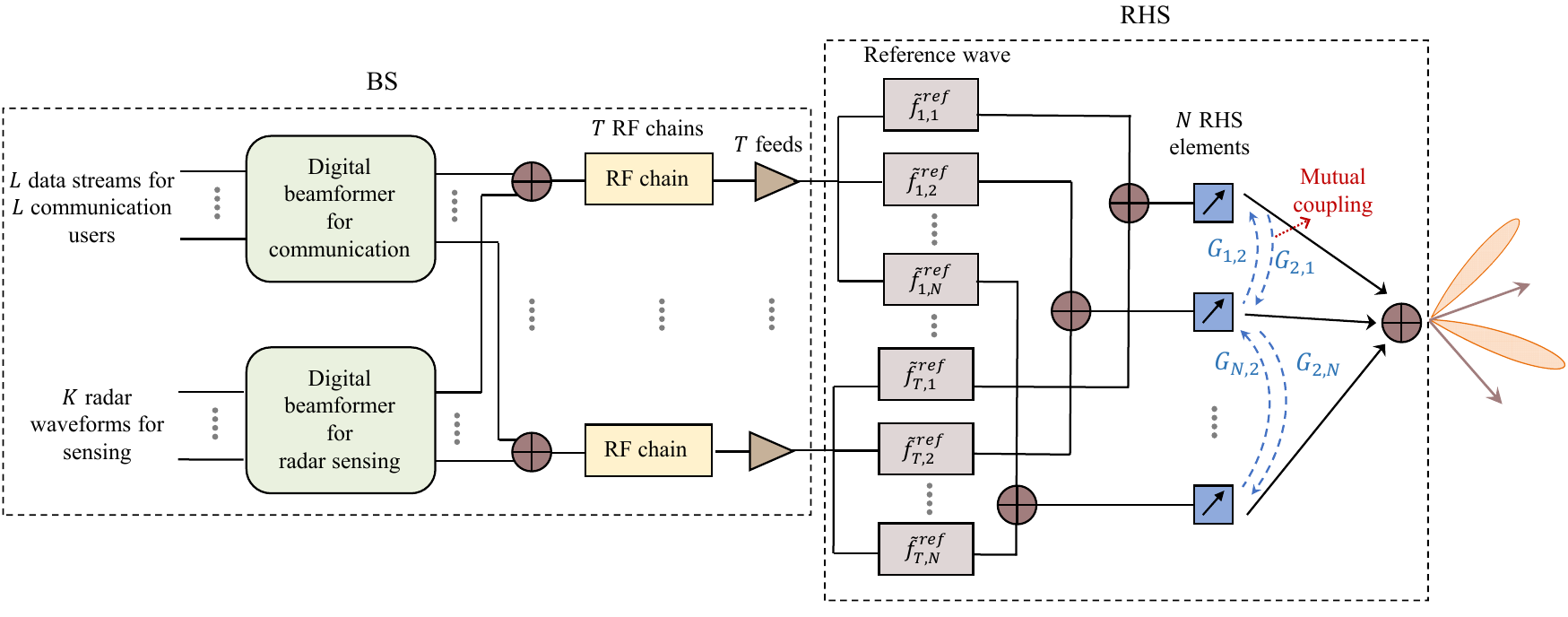}
		%		\vspace{1.5mm}
		\caption{Illustration of the holographic beamforming scheme with mutual coupling effects.}
				\vspace{-3mm}
		\label{beamforming_illustration}
	\end{figure*}
%		\subsubsection{Structure of the Holographic Beamforming Scheme}		
%	The block diagram of the considered holographic beamforming scheme is illustrated in Fig. 3. 
	{As shown in Fig.~\ref{beamforming_illustration}, the holographic beamforming scheme consists of two parts}: the digital beamforming at the BS and the analog beamforming at the RHS. The communication data streams and radar waveforms are first processed by the BS via digital beamforming and are then sent to the RHS for generating the desired ISAC signals via analog beamforming.
		
	Specifically, $L$ data streams and $K$ radar waveforms are first processed by the BS via the digital beamformer for communication $\bm{V}_c \in \mathbb{C}^{T \times L}$ and for radar sensing $\bm{V}_s \in \mathbb{C}^{T \times K}$, respectively. The $L$ data streams carry the information to be sent to $L$ different users, while the $K$ radar waveforms are precoded and sent to the RF chains for radar sensing. Consequently, the signal sent to the $T$ RF chains can be expressed as	
	\begin{equation}
		\label{digital}
%		\bm{x} = \bm{V}_c \bm{c} + \bm{V}_s \bm{s},
\bm{u} = \bm{V}_c \bm{c} + \bm{V}_s \bm{s},
	\end{equation}
	where $\bm{c} = (c_1, \dots, c_L)^{\mathsf{T}} \in \mathbb{C}^{L \times 1}$ denotes the communication symbols intended to the $L$ users, and $\bm{s} = (s_1, \dots, s_K)^{\mathsf{T}} \in \mathbb{C}^{K \times 1}$ denotes the radar waveforms. Each RF chain is connected to one of the feeds in the RHS. The signal $\bm{u}$ is input to the feeds via $T$ RF chains, which convert it from the digital to the RF domain. Then, the signals radiated by the feeds propagate through the waveguide, and finally, they are radiated by the amplitude-modulated radiating elements. Based on (\ref{Holographic_BF}), the object signal generated by the holographic beamforming scheme can be written as		
	\begin{equation}
		\label{tx_signal_RHS}
		\bm{x} = \bm{B} \bm{u}.
	\end{equation}
		
	Without loss of generality, we assume that each element of the communication symbol $\bm{c}$ and each element of the radar waveform $\bm{s}$ have unit power, which can be formulated as $\mathbb{E}(c_l c_l^H) = 1, \forall l$, and $\mathbb{E}(s_k s_k^H) = 1, \forall k$.		
%		\begin{equation}
%			\mathbb{E}(c_l c_l^H) = 1, \forall l,
%		\end{equation}
%		\begin{equation}
%			\mathbb{E}(s_k s_k^H) = 1, \forall k.
%		\end{equation}		
	Furthermore, each communication stream (or radar waveform) is uncorrelated with the other communication streams or radar waveforms, which can be expressed as $\mathbb{E}(c_l c_{l'}^{\mathsf{*}}) = 0, \forall l, l', l \neq l'$, $\mathbb{E}(s_k s_{k'}^*) = 0, \forall k, k', k \neq k'$, and $\mathbb{E}(c_l s_k^*) = 0, \forall l, k$. Based on (\ref{digital}) and (\ref{tx_signal_RHS}), we can derive the data rate model and beampattern model for communication and radar sensing, respectively, as follows.
		
%		\begin{equation}
%			\mathbb{E}(c_l c_{l'}^H) = 0, \forall l, l', l \neq l',
%		\end{equation}
%		\begin{equation}
%			\mathbb{E}(s_k s_{k'}^H) = 0, \forall k, k', k \neq k',
%		\end{equation}
%		\begin{equation}
%			\mathbb{E}(c_l s_k^H) = 0, \forall l, k.
%		\end{equation}
		
		\subsubsection{Data Rate Model for Communication}		
		Let $\bm{h}_l = (h_{1,l}, \dots, h_{N,l})^T$ denote the channel from the RHS to the \newline$l$-th mobile user. Then, the signal received by the $l$-th user can be expressed as		
		\begin{align}
			\label{rec_signal}
			y_l &= \bm{h}_l^{\mathsf{T}} \bm{B} (\bm{V}_c \bm{c} + \bm{V}_s \bm{s} ) + n_l,\notag\\
			&=\bm{h}_l^{\mathsf{T}}\bm{B}\bm{v}_{c,l}c_l+\bm{h}_l^{\mathsf{T}}\bm{B}\sum_{l'\neq l}\bm{v}_{c,l'}c_{l'}+\bm{h}_l^{\mathsf{T}}\bm{B}\bm{V}_s \bm{s}+n_l,
		\end{align}
		where $\bm{v}_{c,l}$ is the $l$-th column of matrix $\bm{V}_c$, and $n_l$ is the noise term which is a Gaussian random variable, i.e., $n_l \sim \mathcal{CN}(0, \sigma^2)$. In (\ref{rec_signal}), the second term and the third term denotes the inter-user interference and the interference caused by the sensing signal, respectively. Therefore, the SINR of the $l$-th user can be expressed as~\cite{Zeng_revisit}		
		\begin{equation}
			\gamma_l = \frac{|\bm{h}_l^{\mathsf{T}}\bm{B}\bm{v}_{c,l}|^2}{|\bm{h}_l^{\mathsf{T}}\bm{B}\sum_{l'\neq l}\bm{v}_{c,l'}|^2+\bm{h}_l^{\mathsf{T}}\bm{B}\bm{V}_s\bm{V}_s^{\mathsf{H}}\bm{B}^{\mathsf{H}}\bm{h}_l^* + \sigma^2 }.
		\end{equation}
		
		Let us define $\bm{V} = (\bm{V}_c, \bm{V}_s)$, and let $\bm{v}_l$ denote the $l$-th column of matrix $\bm{V}$. Then, the SINR of the $l$-th user can be rewritten as		
		\begin{equation}
			\gamma_l = \frac{\bm{h}_l^{\mathsf{T}} \bm{B} \bm{Q}_l \bm{B}^{\mathsf{H}}\bm{h}_l^*}{\bm{h}_l^{\mathsf{T}} \bm{B} \bm{Q} \bm{B}^{\mathsf{H}}\bm{h}_l^*-\bm{h}_l^{\mathsf{T}} \bm{B} \bm{Q}_l \bm{B}^{\mathsf{H}}\bm{h}_l^* + \sigma^2 }.
		\end{equation}		
		where we have introduced the covariance matrices $\bm{Q}_l = \bm{v}_l \bm{v}_l^{\mathsf{H}} \in \mathbb{C}^{T \times T}$ and $\bm{Q} = \bm{V} \bm{V}^{\mathsf{H}} = \sum_{l=1}^{L+K} \bm{v}_l \bm{v}_l^{\mathsf{H}}$.
		
		\subsubsection{Beampattern Model for Radar Sensing}		
		In this subsection, we introduce the far-field power beampattern for radar sensing. {For ease of discussion, Cartesian coordinates are introduced, where the $yoz$ plane coincides with the RHS, and the origin $O$ is located at the center of the RHS. The $x$-axis is perpendicular to the RHS. Under this coordinate system, we use $\theta$ and $\phi$ to represent the elevation and azimuth angles, respectively.} The far-field signal towards the direction $(\theta,\phi)$ can be written as		
		\begin{equation}
			z(\theta, \phi) = \bm{a}^{\mathsf{T}}(\theta, \phi) \bm{B} ( \bm{V}_c \bm{c} + \bm{V}_s \bm{s} ),
		\end{equation}		
		where $\bm{a}(\theta, \phi)=[a_1(\theta,\phi),\cdots,a_n(\theta,\phi),\cdots,a_N(\theta,\phi)]^{\mathsf{T}} \in \mathbb{C}^{N \times 1}$ denotes the steering vector of the RHS. Here, $a_n(\theta, \phi)$ is given by		
		\begin{equation}
			a_n(\theta, \phi) = \exp(-j \bm{k}_f(\theta,\phi)\bm{r}_m^e),
		\end{equation}
		where $\bm{k}_f(\theta,\phi)$ is the propagation vector in free space of the object signal with direction $(\theta,\phi)$. Therefore, the beampattern gain towards the direction $(\theta,\phi)$ is given by
		\begin{equation}
			\label{radiation_pattern}
			P(\theta, \phi) = \mathbb{E} \left(z(\theta, \phi)z^*(\theta, \phi) \right) = \bm{a}^{\mathsf{T}}(\theta, \phi) \bm{B} \bm{Q} \bm{B}^{\mathsf{H}} \bm{a}^*(\theta,\phi).
		\end{equation}
		{Based on (\ref{radiation_pattern}), we can also derive how the reference wave and the geometric channel model are linked to the radiation pattern of the RHS, as discussed in Appendix~\ref{app_radiation_pattern_analysis}.}

	\subsection{Discussions}
%	\subsubsection{Effectiveness of the proposed mutual coupling aware model}
	In the following, we discuss the effectiveness of the proposed EM-compliant communication model to capture the physical behavior of mutual coupling effects in the RHS.
	
%	provide physical insights into the holographic beamformer model in (\ref{Holographic_BF}), explaining how (\ref{Holographic_BF})
	By applying Neumann series to the holographic beamformer in (\ref{Holographic_BF}), the transmitted signal in (\ref{tx_signal_RHS}) can be expressed as
	\begin{align}
		\label{neumman}
		\bm{x}=&k(\sum_{i=0}^{\infty}(e^{j\tau}\bm{\Theta}\bm{G})^i\bm{\Theta})\widetilde{\bm{F}}^{ref}\bm{u},\notag\\
		=&k(e^{j\tau}\bm{\Theta})\widetilde{\bm{F}}^{ref}\bm{u}+k(e^{j\tau}\bm{\Theta})\bm{G}(e^{j\tau}\bm{\Theta})\widetilde{\bm{F}}^{ref}\bm{u}\notag\\
		&+k(e^{j\tau}\bm{\Theta})\bm{G}(e^{j\tau}\bm{\Theta})\bm{G}(e^{j\tau}\bm{\Theta})\widetilde{\bm{F}}^{ref}\bm{u}+\dots\notag\\
		\triangleq&\bm{x}_1+\bm{x}_2+\bm{x}_3+\dots
	\end{align}
	
	Here, the first term $\bm{x}_1\in\mathbb{C}^{N\times 1}$ represents the transmitted signals generated by the time-varying currents excited by the reference wave within the RHS elements.
	%	of the RHS elements excited by the reference wave $\widetilde{\bm{F}}^{ref}$ only, i.e., the reference wave will induce time-variant current in each RHS elements, leading to radiated signals. 
	
	Then, we interpret the physical meaning of the second term $\bm{x}_2\in\mathbb{C}^{N\times 1}$. Specifically, signal $\bm{x}_1$ generated by the RHS elements can propagate to the other RHS elements and generate coupled fields $\bm{G}\bm{x}_1$, where $\bm{G}$ is coupling matrix and characterizes the effect of the propagation environment. The coupled fields will induce additional currents within the RHS elements, which, in turn, will generate new signals, i.e.,~$\bm{x}_2$. 
	
	%	where the influence of propagation environments can be characterized by $\bm{G}$, they will induce additional currents in these RHS elements. The additional currents will further generate new electromagnetic signals, as represented by $\bm{x}_2$. 
	Similarly, signal $\bm{x}_2$ generated by the RHS elements will further couple to other RHS elements, inducing new signals~$\bm{x}_3$. 
	%	After that, such generated signals will further interact with the other RHS elements and induce new transmitted signals, as indicated by the third term in (\ref{neumman}). 
	
	This mutual interaction process continues indefinitely until convergence, and eventually the EM signals that are emitted out from the transmit RHS antenna and fed into the wireless channels, i.e., $\bm{x}$, are generated by the superposition of all these signals. The above discussion are summarized in the following remark. 
	\begin{remark}
		The proposed holographic beamformer in (\ref{Holographic_BF}) can effectively capture the mutual coupling effect in the RHS.
	\end{remark}

	\section{Impact of Mutual Coupling on ISAC Performance}
	\label{subsection_impact}
	Based on the derived mutual coupling aware holographic ISAC model in the previous section, we first analyze the impact of mutual coupling on the beampatterns generated by the RHS through numerical simulations in this section. Then, we provide a mathematical explanation of such impact.
	
%	Based on the mutual coupling aware holographic ISAC model derived in the previous section, we first illustrate the impact of mutual coupling on the beam generated by the RHS through numerical simulations in this section. Subsequently, we provide a mathematical explanation of this effect.
	
	Specifically, we assume that an ISAC BS adopts an RHS array with $20\times 20$ elements and $T=3$ feeds. Further, we assume that there are one communication users and two sensing directions. The setting of the reference wave $\widetilde{\bm{F}}^{ref}$, the maximum and minimum magnitude of the polarizability $\bm{\Theta}$, the coupling matrix $\bm{G}$, and other related parameters can be found in the simulation part in Section~\ref{sec_simulation}. Existing holographic beamforming scheme in~\cite{Zhang_HISAC_2022} is adopted to configure the RHS, where the mutual coupling effect is not considered since ISAC performance is not sensitive to sidelobe level in the work. The beampattern generated by the RHS is given in Fig.~\ref{Superposability}. From Fig.~\ref{Superposability}, we can see that due to the mutual coupling effect, the RHS can generate notable sidelobes with a sidelobe level~(SLL) of $-1.88$~dB. The sidelobe may point towards environmental clutterers, cause interference to sensing procedure and deteriorating sensing performance.
%	Further, the mutual coupling effect also leads to server sidelobes, with a sidelobe level~(SLL) of $-0.54$~dB. The beam deviation and severe sidelobes jointly result in lower beamforming gain, thus degrading both the communication and sensing performances of ISAC systems.
	
%	Further, the generated beam deviates from the target sensing direction by, which thus can degrade the performance of ISAC systems.

	{
	\begin{figure}[!t]
		\centering
		\includegraphics[width=0.43\textwidth]{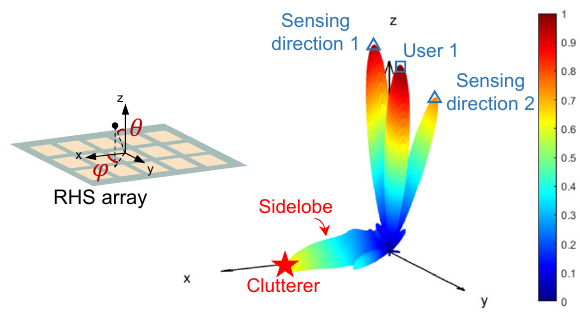}
		%		\vspace{1.5mm}
		\caption{Beampattern of an RHS array when mutual coupling effects are not considered in holographic beamforming~\cite{Deng_VTM_2023}. One mobile user and two sensing directions are considered, with one clutterer in the environment. The number of RHS elements is set as $N=20^2$.}
		\vspace{-3mm}
		\label{Superposability}
	\end{figure}}

	Further, we try to explain the impact of mutual coupling mathematically. According to (\ref{ideal_Holographic_BF}), for an ideal RHS without mutual coupling, the transmitted signals of the RHS are
	\begin{align}
		\label{ideal_transmit_signal}
		\bm{x}=\bm{M}\widetilde{\bm{F}}^{ref}\bm{u},
	\end{align}
	where $\bm{u}$ is the signal vector input into different feeds of the RHS, $\widetilde{\bm{F}}^{ref}$ is reference signals generated by different feeds of the RHS, and $\bm{M}$ is the holographic pattern of the RHS. On the other hand, for an RHS with mutual coupling effects, the transmitted signals of the RHS in (\ref{neumman}) can be reformulated as
{
	\begin{align}
		\label{neumman_rewritten}
		\bm{x}=&k\Big(e^{j\tau}\bm{\Theta})(\widetilde{\bm{F}}^{ref}\!+\!\bm{G}(e^{j\tau}\bm{\Theta})\widetilde{\bm{F}}^{ref}\!\notag\\
		&\quad+\bm{G}(e^{j\tau}\bm{\Theta})\bm{G}(e^{j\tau}\bm{\Theta})\widetilde{\bm{F}}^{ref}\!+\!\dots\Big)\bm{u},\notag\\
		\triangleq& k(e^{j\tau}\bm{\Theta})(\widetilde{\bm{F}}^{ref})^{equ}\bm{u}.
	\end{align}}

	By comparing (\ref{neumman_rewritten}) with (\ref{ideal_transmit_signal}), we can understand $(\widetilde{\bm{F}}^{ref})^{equ}$ as ``equivalent reference wave" when mutual coupling effects are present, which is the superposition of the original reference wave $\widetilde{\bm{F}}^{ref}$ and coupled fields $\sum_{i=1}^{\infty}\left(\bm{G}(e^{j\tau}\bm{\Theta})\right)^i\widetilde{\bm{F}}^{ref}$. Since the equivalent reference wave distribution is different from the desired  original reference wave distribution, the generated beam of the RHS can be distorted. The above discussions are summarized into the following remark.
	\begin{remark}
		Due to the mutual coupling effect within RHS, coupled fields will superpose with original reference waves to form equivalent reference waves. Since the equivalent reference wave distributions are different from the original desired reference wave distributions, the generated beam of the RHS is distorted, generating notable undesired sidelobes.
%		that excite different RHS elements. 
	\end{remark}
%	(explanation, why causing the deviation: the change of reference wave.)
	
	Therefore, it is critical to revisit the holographic beamforming problem for ISAC systems and design an efficient algorithm to eliminate sidelobes and improve ISAC performance, as shown in the following sections. 
	\vspace{-2mm}
	\section{Holographic ISAC Problem Formulation and Decomposition}
	\label{sec_problem_formulation}		
	\vspace{-1mm}
		In this section, to compensate for the performance degradation caused by the mutual coupling effect, we first formulate a holographic ISAC problem by taking the mutual coupling effect into account. Then, to efficiently solve the formulated problem, we decouple it into two subproblems, i.e., digital beamforming subproblem and holographic beamforming subproblem.
		\vspace{-2mm}
		\subsection{Problem Formulation}		
		The aim of the holographic ISAC problem is to improve communication performance subject to specified radar sensing and power constraints by jointly optimizing the digital beamformer $\bm{V}$ and holographic pattern $\bm{\Theta}$. To guarantee the user fairness, the minimum data rate among mobile users are selected as the objective function. Mathematically, the optimization problem can be formulated as:		
{\setlength{\abovedisplayskip}{2pt}
	\setlength{\belowdisplayskip}{3pt}
\begin{subequations}\label{opt_problem}
	\begin{align}
		\label{obj}
		&\max_{\bm{\Theta}, \bm{V}}  \min_{l} R_l \hfill\\
		\label{cons_max_sensing}
		s.t.~& P(\theta_w, \phi_w)\le G^{th}_w,\quad w = 1, \dots, W,\hfill\\
		\label{cons_fair_sensing}
		& \gamma_d^l P(\theta_1, \phi_1) \le  P(\theta_d, \phi_d) \le \gamma_d^u P(\theta_1, \phi_1), \quad d = 2, \dots, D, \hfill\\
		\label{cons_min_sensing}
		& P(\theta_d, \phi_d)\ge G^{th}_d,\quad d = 1, \dots, D,\hfill\\
		\label{cons_min_rate}
		& R_l\ge R^{th}_l,\quad l=1,\dots,L,\hfill\\
		\label{cons_max_transmit}
		& \Tr(\bm{V}\bm{V}^H)\le P_M,\hfill\\
		\label{cons_pol_range}
		& \theta_n\in[\theta_{min},\theta_{max}],\hfill
	\end{align}
\end{subequations}}		
		where $G^{th}_w$ is the maximum allowed beampattern gain in the direction of the $w$-th clutterer, \( \gamma_d^l \) and \( \gamma_d^u \) are two positive parameters controlling the beampattern gains at different directions, $G^{th}_d$ is the minimum required beampattern gain in the $d$-th angular region of interest, \( R_l^{th} \) is the minimum required data rate of the \( l \)-th user, and \( P_M \) is maximum transmit power. 
		
		{Constraint (\ref{cons_max_sensing}) limits the sidelobe levels in the direction \( (\theta_w, \phi_w) \), which cannot exceed threshold $G_w^{th}$. Constraint~(\ref{cons_fair_sensing}) ensures that the beampattern gains toward the directions \( (\theta_d, \phi_d) \) are within the range \( [\gamma_d^l P(\theta_1, \phi_1), \gamma_d^u P(\theta_1, \phi_1)] \), which allows us to control the beampattern gains towards different targets\footnote{The maximum allowed sidelobe level $G^{th}_w$ is predetermined, which is mainly dependent on the received SINR threshold for correct detection, the location and RCS of the $w$-th clutterer, the number $W$ of clutterers, noise variance $\sigma^2$, and the power of received desired sensing signals. Further, the parameters $\gamma_d^l$ and $\gamma_d^u$ are also pre-defined, where in a real system, we select values approaching $1$ for them so as to balance the detection accuracy in different targeted sensing directions.}.} Constraint (\ref{cons_min_sensing}) sets beampattern gain requirements for each angular region of interest\footnote{As noted in~\cite{Stoica_probing_2007}, the radar performance is related to the beampattern gains towards the directions of the targets. Specifically, by increasing the beampattern gains in correspondence of the direction of the targets, the SINR of the signals reflected by the targets (i.e., the echo signals) can be enhanced, leading to higher sensing performance.} Constraint (\ref{cons_min_rate}) guarantees the performance for each communication user. Constraint (\ref{cons_max_transmit}) is the sum-power constraint, and (\ref{cons_pol_range}) accounts for the value range for the polarizability magnitude of each RHS element.

		\subsection{Problem Decomposition}		
%		Based on Remark~\ref{remark_compare_ideal}, we can find that the mutual coupling effect introduces the additional coupling matrix $\bm{G}$ and inverse matrix operation into the holographic beamformer, which brings non-linearity and complicates the holographic beamforming design. 
		Since the digital beamformer \( \bm{V} \) and the holographic pattern $\bm{\Theta}$ are coupled, which further makes the problem difficult to solve directly. To efficiently address this, we decompose the problem into the following two subproblems, optimizing \( \bm{V} \) and $\bm{\Theta}$ separately.
	\subsubsection{Digital Beamforming} 
	Given the holographic pattern $\bm{\Theta}$, the digital beamforming subproblem can be given by
	\begin{subequations}\label{opt_problem_DBF}
		\begin{align}
			\label{obj_DBF}
			&\max_{\bm{V}}  \min_{l} R_l \hfill\\
			\label{cons_max_sensing_DBF}
			s.t.~&P(\theta_w, \phi_w)\le G^{th}_w,\quad w = 1, \dots, W,\hfill\\
			\label{cons_fair_sensing_DBF}
			& \gamma_d^l P(\theta_1, \phi_1) \le  P(\theta_d, \phi_d) \le \gamma_d^u P(\theta_1, \phi_1), \quad d = 2, \dots, D, \hfill\\
			\label{cons_min_sensing_DBF}
			& P(\theta_d, \phi_d)\ge G^{th}_d,\quad d = 1, \dots, D,\hfill\\
			\label{cons_min_rate_DBF}
			& R_l\ge R^{th}_l,\quad l=1,\dots,L,\hfill\\
			\label{cons_max_transmit_DBF}
			& \Tr(\bm{V}\bm{V}^H)\le P_M.\hfill
		\end{align}
	\end{subequations}
	
	\subsubsection{Holographic Beamforming}
	Given the digital beamformer $\bm{V}$, the holographic beamforming subproblem can be given by
	\begin{subequations}\label{opt_problem_HBF}
		\begin{align}
			\label{obj_HBF}
			&\max_{\bm{\Theta}}  \min_{l} R_l \hfill\\
			\label{cons_max_sensing_HBF}
			s.t.~&P(\theta_w, \phi_w)\le G^{th}_w,\quad w = 1, \dots, W,\hfill\\
			\label{cons_fair_sensing_HBF}
			& \gamma_d^l P(\theta_1, \phi_1) \le  P(\theta_d, \phi_d) \le \gamma_d^u P(\theta_1, \phi_1), \quad d = 2, \dots, D, \hfill\\
			\label{cons_min_sensing_HBF}
			& P(\theta_d, \phi_d)\ge G^{th}_d,\quad d = 1, \dots, D,\hfill\\
			\label{cons_min_rate_HBF}
			& R_l\ge R^{th}_l,\quad l=1,\dots,L,\hfill\\
			\label{cons_pol_range_HBF}
			& \theta_n\in[\theta_{min},\theta_{max}].\hfill
		\end{align}
	\end{subequations}

\section{Mutual Coupling Aware Beamforming Algorithm Design}
\label{sec_optimization_algorithm}
In this section, we design algorithms to solve the two subproblems (\ref{opt_problem_DBF}) and (\ref{opt_problem_HBF}), respectively, based on which an overall mutual coupling aware algorithm is proposed to solve the holographic ISAC problem in (\ref{opt_problem}). 
\subsection{Digital Beamforming Design}
To solve the non-convex problem of digital beamforming, we convert it into a series of  quadratically constrained quadratic programs~(QCQP) through fractional programming, which can be effectively solved by applying semidefinite relaxation techniques~(SDR) techniques.
\subsubsection{Problem Reformulation}
\label{subsubsection_DB_reformulation}
The key idea is to eliminate the minimization operation in the objective function (\ref{obj_DBF}) and the non-convex ratio in the expression of $R_l$, i.e., $\frac{\bm{h}_l^T \bm{B} \bm{Q}_l \bm{B}^H\bm{h}_l^*}{\bm{h}_l^T \bm{B} \bm{Q} \bm{B}^H\bm{h}_l^*-\bm{h}_l^T \bm{B} \bm{Q}_l \bm{B}^H\bm{h}_l^* + \sigma^2 }$ based on fractional programming technique.

Specifically, since the data rate $R_l$ of the $l$-th user increases with the received SNR $\gamma_l$, the maximization of the minimum data rate among the users is equivalent to the maximization of the minimum received SNR among the users, i.e., the digital beamforming problem in (\ref{opt_problem_DBF}) is equivalent to
\begin{subequations}\label{opt_problem_DBF_v2}
	\begin{align}
		\label{obj_DBF_v2}
		&\max_{\bm{V}}  \min_{l} \gamma_l \hfill\\
		s.t.~&(\ref{cons_max_sensing_DBF})-(\ref{cons_max_transmit_DBF}).\hfill \notag
	\end{align}
\end{subequations}

To eliminate the minimization operation in the objective function, we introduce auxiliary variable $\tau$,~i.e.,
\begin{subequations}\label{opt_problem_DBF_v3}
	\begin{align}
		\label{obj_DBF_v3}
		&\max_{\bm{V},\tau}  \tau  \hfill\\
		\label{cons_aux_DBF_v3}
		s.t.~&\tau\le \gamma_l,~l=1,\dots,L,\hfill\\
		&(\ref{cons_max_sensing_DBF})-(\ref{cons_max_transmit_DBF}),\hfill \notag
	\end{align}
\end{subequations}
 where constraint (\ref{cons_aux_DBF_v3}) is introduced to force the auxiliary variable $\tau$ to approach the minimum SNR among the users. 
 
 Note that constraint (\ref{cons_aux_DBF_v3}) is  non-convex due to the non-convex ratio in the expression of $\gamma_l$. To address this, we apply the fractional programming technique, and thus problem (\ref{opt_problem_DBF_v3}) can be equivalently transformed into
% \begin{subequations}\label{opt_problem_DBF_v4}
% 	\begin{align}
% 		\label{obj_DBF_v4}
% 		&\max_{\bm{V},\tau,\{\rho_l\}_l}  \tau  \hfill\\
% 		\label{cons_aux_DBF_v4}
% 		s.t.~&\tau\le 2\Re(\rho_l^*\bm{v}_{c,l}^H\bm{B}^H\bm{h}_l^*)-\notag\hfill\\
% 		&\quad~~|\rho_l|^2(\bm{h}_l^T \bm{B} \bm{Q} \bm{B}^H\bm{h}_l^*-\bm{h}_l^T \bm{B} \bm{Q}_l \bm{B}^H\bm{h}_l^* + \sigma^2),\hfill\\
% 		&(\ref{cons_fair_sensing_DBF})-(\ref{cons_max_transmit_DBF}),\hfill \notag
% 	\end{align}
% \end{subequations}
\begin{subequations}\label{opt_problem_DBF_v4}
	\begin{align}
		\label{obj_DBF_v4}
		&\max_{\bm{V},\tau,\{\rho_l\}_l}  \tau  \hfill\\
		\label{cons_aux_DBF_v4}
		s.t.~&\tau\le-|\rho_l|^2(\bm{h}_l^T \bm{B} \bm{Q} \bm{B}^H\bm{h}_l^*-\bm{h}_l^T \bm{B} \bm{Q}_l \bm{B}^H\bm{h}_l^* + \sigma^2) \notag\hfill\\
		&\quad~~+2\Re(\rho_l^*\bm{v}_{c,l}^H\bm{B}^H\bm{h}_l^*),\hfill\\
		&(\ref{cons_max_sensing_DBF})-(\ref{cons_max_transmit_DBF}),\hfill \notag
	\end{align}
\end{subequations}
 To solve problem (\ref{opt_problem_DBF_v4}), we perform iterative optimization over $\{\rho_l\}_l$, $\bm{V}$, and $\tau$. The derivation of optimal solutions with respect to these variables is illustrated below.
 
 \subsubsection{Digital Beamforming Design}
 Given variables $\bm{V}$, we first jointly optimize $\{\rho_l\}_l$ and $\tau$, with the optimal $\rho_l^{opt}$ given by~\cite{Shen_load_spawc}
 \begin{align}
 	\label{rho_opt}
 	\rho_l^{opt}=\frac{\bm{v}_{c,l}^{\mathsf{H}}\bm{B}^{\mathsf{H}}\bm{h}_l^*}{\bm{h}_l^{\mathsf{T}} \bm{B} \bm{Q} \bm{B}^{\mathsf{H}}\bm{h}_l^*-\bm{h}_l^{\mathsf{T}} \bm{B} \bm{Q}_l \bm{B}^{\mathsf{H}}\bm{h}_l^* + \sigma^2}.
 \end{align}
 
 Then, give the auxiliary variables $\{\rho_l\}_l$, the digital beamformer can be optimized by solving the following problem,
 {\setlength{\abovedisplayskip}{-6pt}
 	\setlength{\belowdisplayskip}{3pt}
 \begin{subequations}\label{opt_problem_DBF_v5}
 	\begin{align}
 		\label{obj_DBF_v5}
 		&\max_{\bm{V},\tau}  \tau  \hfill\\
 		s.t.~&(\ref{cons_aux_DBF_v4}),~(\ref{cons_max_sensing_DBF})-(\ref{cons_max_transmit_DBF}).\hfill \notag
 	\end{align}
 \end{subequations}}
 
 Since problem (\ref{opt_problem_DBF_v5}) is a quadratic constrained quadratic programs~(QCQP), we can apply the SDR technique to solve it. Specifically, since the right-hand-side of constraint (\ref{cons_aux_DBF_v4}) is a quadratic function with respect to digital beamforming vector $\bm{v}_{l,c}$ with linear terms, we introduce auxiliary variable $t$ into $\bm{v}_{l,c}$ to form $\bm{\widetilde{v}}_{l,c}=[\bm{v}_l^{\mathsf{T}},t]^{\mathsf{T}}$ so as to remove the linear term. Then, we can see that the constraints are determined by the covariance matrices $\bm{\widetilde{Q}}_l = \bm{\widetilde{v}}_l\bm{\widetilde{v}}_l^{\mathsf{H}}$ and $\bm{Q} = \bm{V} \bm{V}^{\mathsf{H}}$. Therefore, problem (\ref{opt_problem_DBF_v5}) can be reformulated, as a function of these latter matrices, in the equivalent problem:
%  and by introducing covariance matrix   $\bm{\widetilde{Q}}_l = \bm{\widetilde{v}}_l\bm{\widetilde{v}}_l^H$, we can find that constraint (\ref{cons_aux_DBF_v4}) is linear with respect to $\bm{\widetilde{Q}}_l$. Further, all the constraints are linear with respect
% note that the right-hand-side of constraint (\ref{cons_aux_DBF_v4}) is a quadratic function with respect to $\bm{v}_{l,c}$ with linear terms. Therefore, by introducing auxiliary variable $t$ into the digital beamforming vector $\bm{v}_{l,c}$ to derive $\bm{\widetilde{v}}_{l,c}=[\bm{v}_l^T,t]^T$, and by introducing covariance matrix   $\bm{\widetilde{Q}}_l = \bm{\widetilde{v}}_l\bm{\widetilde{v}}_l^H$, we can find that constraint (\ref{cons_aux_DBF_v4}) is linear with respect to $\bm{\widetilde{Q}}_l$. Further, all the constraints are linear with respect 
 {\setlength{\abovedisplayskip}{-0pt}
 	\setlength{\belowdisplayskip}{3pt}
 \begin{subequations}\label{opt_problem_DBF_v6}
 	\begin{align}
 		\label{obj_DBF_v6}
 		&\max_{\bm{Q},\{\bm{\widetilde{Q}}_l\}_l,\tau}  \tau  \hfill\\
 		\label{cons_auxiliary_v6}
 		s.t.~&\tau\le\Tr(\bm{\widetilde{Q}}_l\bm{\widetilde{C}}_l)-\Tr(\bm{Q}(|\rho_l|^2\bm{A}_l))-|\rho_l|^2\sigma^2, l=1,\dots,L,\hfill\\
 		\label{cons_max_sensing_v6}
 		&\Tr(\bm{Q}\bm{A}(\theta_w,\phi_w))\le G_w^{th},w=1,\dots,W,\hfill\\
 		\label{cons_pattern_gain_lb_v6}
 		&\Tr(\bm{Q}\bm{A}(\theta_d,\phi_d))-\gamma_d^l\Tr(\bm{Q}\bm{A}(\theta_1,\phi_1))\ge 0,d=2,\dots,D,\hfill\\
 		\label{cons_pattern_gain_ub_v6}
 		&\Tr(\bm{Q}\bm{A}(\theta_d,\phi_d))-\gamma_d^u\Tr(\bm{Q}\bm{A}(\theta_1,\phi_1))\le 0,d=2,\dots,D,\hfill\\
 		\label{cons_pattern_gain_threshold_v6}
 		&\Tr(\bm{Q}\bm{A}(\theta_d,\phi_d))\ge G_d^{th},d=1,\dots,D,\hfill\\
 		\label{cons_rate_threshold_v6}
 		&(\frac{1}{2^{R_l^{th}}-1}+1)\Tr(\bm{\widetilde{Q}}_l\bm{\widetilde{A}}_l)-\Tr(\bm{Q}\bm{A}_l)\ge \sigma^2,~l=1,\dots,L\hfill\\
 		\label{cons_max_transmit_power_v6}
 		&\Tr(\bm{Q})\le P_M,\hfill\\
 		\label{cons_additional_t_v6}
 		&\Tr(\bm{\widetilde{Q}}_lZ)=1,\hfill\\
 		\label{cons_Q_l_semidefinite_v6}
 		&\bm{\widetilde{Q}}_l\succeq 0, l=1,\dots,L,\hfill\\
 		\label{cons_Q_l_rank_1_v6}
 		&\rank(\bm{\widetilde{Q}}_l)=1,l=1,\dots,L,\hfill\\
 		\label{cons_Q_semidefinite_v6}
 		&\bm{Q}\succeq 0,\hfill\\
 		\label{cons_Q_sensing_semidefinite_v6}
 		&\bm{Q}-\sum_{l=1}^L\bm{E}\bm{\widetilde{Q}}_l\bm{E}^{\mathsf{T}}\succeq 0.\hfill
 	\end{align}
 \end{subequations}}
 Here, constraint (\ref{cons_pattern_gain_lb_v6}) and (\ref{cons_pattern_gain_ub_v6}) correspond to constraint (\ref{cons_fair_sensing_DBF}). Further, constraint (\ref{cons_auxiliary_v6}), (\ref{cons_max_sensing_v6}), (\ref{cons_pattern_gain_threshold_v6}), (\ref{cons_rate_threshold_v6}), and (\ref{cons_max_transmit_power_v6}) are linear formulations of the original constraints (\ref{cons_aux_DBF_v4}), (\ref{cons_max_sensing_DBF}), (\ref{cons_min_sensing_DBF}), (\ref{cons_min_rate_DBF}), and (\ref{cons_max_transmit_DBF}), respectively. Constraint (\ref{cons_Q_l_semidefinite_v6}) and (\ref{cons_Q_semidefinite_v6}) ensure that both matrices $\bm{\widetilde{Q}}_l$ and $\bm{Q}$ are positive semidefinite. Further, constraint (\ref{cons_Q_l_semidefinite_v6}) and (\ref{cons_Q_l_rank_1_v6}) guarantee that the matrix $\bm{\widetilde{Q}}_l$ can be decomposed as $\bm{\widetilde{Q}}_l=\bm{\widetilde{v}}_l\bm{\widetilde{v}}_l^{\mathsf{H}}$. Further, constraint (\ref{cons_Q_sensing_semidefinite_v6}) guarantees that matrix $\bm{Q}-\sum_{l=1}^L\bm{E}\bm{\widetilde{Q}}_l\bm{E}^{\mathsf{T}}$ can be decomposed as $\bm{V}_s\bm{V}_s^{\mathsf{H}}$. In addition, constraint (\ref{cons_additional_t_v6}) is a linear formulation of the additional constraint brought by the introduction of auxiliary variable $t$, i.e., $tt^*=1$. Also, we have utilized the following definitions:
 \begin{align}
 	\bm{A}(\theta_d,\phi_d)&=\bm{B}^{\mathsf{H}} \bm{a}^*(\theta,\phi)\bm{a}^{\mathsf{T}}(\theta, \phi) \bm{B},\\
 	\bm{A}_l&=\bm{B}^{\mathsf{H}}\bm{h}_l^*\bm{h}_l^{\mathsf{T}} \bm{B},\\
 	\bm{\widetilde{A}}_l &= 
 	\begin{bmatrix}
 		\bm{A}_l & \bm{0}_{K\times1} \\
 		\bm{0}_{1\times K} & 0
 	\end{bmatrix},\\
 	\bm{\widetilde{C}}_l&=
 	\begin{bmatrix}
 		|\rho_l|^2\bm{A}_l & \rho_l^*\bm{B}^{\mathsf{H}}\bm{h}_l^* \\
 		\rho_l\bm{h}_l^{\mathsf{T}}\bm{B} & 0
 	\end{bmatrix},\\
 	\bm{E}&=[\bm{I}_{K\times K},\bm{0}_{K\times 1}],\\
% 	\begin{bmatrix}
% 		0 & \cdots  & 0 \\
% 		\vdots &   & \vdots \\
% 		0 & \cdots  & 1
% 	\end{bmatrix}
	\bm{Z}&=\begin{bmatrix}
		0 & \cdots & & & 0 \\
		\vdots & \ddots & & & \vdots \\
		& & 0 & & \\
		\vdots & & & \ddots & \vdots \\
		0 & \cdots & & & 1
	\end{bmatrix}.
 \end{align}
 
 Problem (\ref{opt_problem_DBF_v6}) is still non-convex due to the rank-one constraint (\ref{cons_Q_l_rank_1_v6}). According to the SDR technique, we first drop the rank   constraint (\ref{cons_Q_l_rank_1_v6}) to obtain a relaxed version of problem (\ref{opt_problem_DBF_v6}), which is convex. The relaxed convex optimization problem can thus be solved effectively by using interior point methods.   Let us denote by $\bm{Q}^{rlx}, \bm{\widetilde{Q}}_1^{rlx}, \cdots, \bm{\widetilde{Q}}_L^{rlx}$ the optimal solution of the relaxed convex optimization problem. Similar to the proof in \cite{Haobo_HISAC_2022}, a rank-one solution $\bm{\widetilde{Q}}_1,\cdots, \bm{\widetilde{Q}}_L$ of problem (\ref{opt_problem_DBF_v6}) can be obtained as the following theorem.
% To satisfy the rank-one constraint (\ref{cons_Q_l_rank_1_v6}), we select can recover 
 \begin{theorem}
 	\label{prop_recover_rank_one_solution}
 	Let $\bm{Q}^{rlx}, \bm{\widetilde{Q}}_1^{rlx}, \cdots, \bm{\widetilde{Q}}_L^{rlx}$ denote the optimal solution to the relaxed problem obtained from (\ref{opt_problem_DBF_v6}). A solution of problem (\ref{opt_problem_DBF_v6}) can be given by
 	\begin{align}
 		\bm{Q}&=\bm{Q}^{rlx},\\
 		\bm{\widetilde{Q}}_l&=
 		\begin{bmatrix}
 			\frac{\bm{v}_{l,c}^{rlx}(\bm{v}_{l,c}^{rlx})^{\mathsf{H}}\bm{B}^{\mathsf{H}}\bm{h}_l^*\bm{h}_l^{\mathsf{T}}\bm{B}\bm{v}_{l,c}^{rlx}(\bm{v}_{l,c}^{rlx})^{\mathsf{H}}}{\bm{h}_l^{\mathsf{T}}\bm{B}\bm{v}_{l,c}^{rlx}(\bm{v}_{l,c}^{rlx})^{\mathsf{H}}\bm{B}^{\mathsf{H}}\bm{h}_l^*} & \frac{\bm{v}_{l,c}^{rlx}(\bm{v}_{l,c}^{rlx})^{\mathsf{H}}\bm{B}^{\mathsf{H}}\bm{h}_l^*}{|\bm{h}_l^{\mathsf{T}}\bm{B}\bm{v}_{l,c}^{rlx}(\bm{v}_{l,c}^{rlx})^{\mathsf{H}}|} \\
 			\frac{\bm{h}_l^{\mathsf{T}}\bm{B}\bm{v}_{l,c}^{rlx}(\bm{v}_{l,c}^{rlx})^{\mathsf{H}}}{|\bm{h}_l^{\mathsf{T}}\bm{B}\bm{v}_{l,c}^{rlx}(\bm{v}_{l,c}^{rlx})^{\mathsf{H}}|} & 1
 		\end{bmatrix}\\
 		\tau&=\min_l\Tr(\bm{\widetilde{Q}}_l\bm{\widetilde{C}}_l)-\Tr(\bm{Q}(|\rho_l|^2\bm{A}_l))-|\rho_l|^2\sigma^2,
 	\end{align}
 	where $\bm{\widetilde{v}}_{l,c}^{rlx}=\xi_l^{max}\bm{\eta}_l$. Here, $\xi_l^{max}$ is the maximum eigenvalue of matrix $\bm{\widetilde{Q}}_l$ and $\bm{\eta}_l$ is the corresponding eigenvector. Then, denote $\bm{v}_{l,c}^{rlx}$ as the vector consisting of the first $T$ elements of vector $\bm{\widetilde{v}}_{l,c}^{rlx}$, i.e., $\bm{v}_{l,c}^{rlx}=\bm{\widetilde{v}}_{l,c}^{rlx}(1:T)$.
% 	 Further, define $\bm{v}_{l,c}=\frac{\bm{h}_l^T\bm{B}\bm{v}_{l,c}^{rlx}(\bm{v}_{l,c}^{rlx})^H}{|\bm{h}_l^T\bm{B}\bm{v}_{l,c}^{rlx}(\bm{v}_{l,c}^{rlx})^H|}$.
 \end{theorem}

\subsection{Holographic Beamforming Design}
Recall the holographic beamformer derived in (\ref{Holographic_BF}). Due to the non-linearity brought by the coupling intensity matrix $\bm{G}$, the holographic beamforming subproblem (\ref{opt_problem_HBF}) is non-convex and thus is difficult to solve. To tackle this issue, our main idea is to convert it into a series of QCQP sub-problems based on Neuman series, which can be tackled through SDR techniques.

Similar to Section~\ref{subsubsection_DB_reformulation}, we first eliminate the minimization operation and the non-convex ratio in the objective function (\ref{obj_HBF}) through fractional programming, where the holographic beamforming problem in (\ref{opt_problem_HBF}) is reformulated as
\begin{subequations}\label{opt_v2_problem_HBF}
	\begin{align}
		\label{obj_v2_HBF}
		&\max_{\bm{\Theta},\tau,\{\rho_l\}_l}  \tau  \hfill\\
		s.t.~&\tau\le-|\rho_l|^2(\bm{h}_l^{\mathsf{T}} \bm{B} \bm{Q} \bm{B}^{\mathsf{H}}\bm{h}_l^*-\bm{h}_l^{\mathsf{T}} \bm{B} \bm{Q}_l \bm{B}^{\mathsf{H}}\bm{h}_l^* + \sigma^2) \notag\hfill\\
		\label{cons_v2_aux_HBF}
		&\quad~~+2\Re(\rho_l^*\bm{v}_{c,l}^{\mathsf{H}}\bm{B}^{\mathsf{H}}\bm{h}_l^*),\hfill\\
		&(\ref{cons_max_sensing_HBF})-(\ref{cons_pol_range_HBF}).\hfill \notag
	\end{align}
\end{subequations}

Similarly, we can solve problem (\ref{opt_v2_problem_HBF}) by performing iterative optimizations over variables $\{\rho_l\}_l$, $\bm{\Theta}$, and $\tau$. Specifically, given the holographic pattern $\bm{\Theta}$, we first optimize $\{\rho_l\}_l$ and $\tau$ jointly, where the optimal $\rho_l^{opt}$ can be given by
\begin{align}
	\label{rho_opt_HBF}
	\rho_l^{opt}=\frac{\bm{v}_{c,l}^{\mathsf{H}}\bm{B}^{\mathsf{H}}\bm{h}_l^*}{\bm{h}_l^{\mathsf{T}} \bm{B} \bm{Q} \bm{B}^{\mathsf{H}}\bm{h}_l^*-\bm{h}_l^{\mathsf{T}} \bm{B} \bm{Q}_l \bm{B}^{\mathsf{H}}\bm{h}_l^* + \sigma^2}.
\end{align}

Then, given the auxiliary variables $\{\rho_l\}_l$, the holographic pattern $\bm{\Theta}$ can be optimized by solving the following problem
{
	\setlength{\abovedisplayskip}{-6pt}
	\setlength{\belowdisplayskip}{3pt}
\begin{subequations}\label{opt_v3_problem_HBF}
	\begin{align}
		\label{obj_v3_HBF}
		&\max_{\bm{\Theta},\tau}  \tau  \hfill\\
		s.t.~&(\ref{cons_v2_aux_HBF}), (\ref{cons_max_sensing_HBF})-(\ref{cons_pol_range_HBF}),\hfill \notag
	\end{align}
\end{subequations}}
We can see that the holographic pattern $\bm{\Theta}$ has an influence on the constraint functions of problem (\ref{opt_v3_problem_HBF}) through holographic beamformer $\bm{B}$. According to (\ref{Holographic_BF}), due to the mutual coupling effect, the holographic pattern $\bm{\Theta}$ is related to the holographic beamformer $\bm{B}$ through inverse matrix operations, which brings non-linearity to holographic beamforming problem, making it inherently non-convex and thus challenging to solve. 

To cope with this issue, we can make use of the Neumann series to get rid of the inverse matrix operations. The main idea is that in each iteration, we only allow small changes to the holographic pattern $\bm{\Theta}$, which can be denoted by $\delta_k\widetilde{\bm{\Theta}}$. Therefore, the holographic beamformer $\bm{B}$ can be reformulated as a linear function with respect to the small increment $\delta\widetilde{\bm{\Theta}}$ by utilizing the Neumann series.

Specifically, in the $k$-th iteration, we can update the holographic pattern in the following manner,
\begin{align}
	\label{delta_Theta}
	\bm{\Theta}_{k}^{-1}=\bm{\Theta}_{k-1}^{-1}-\delta_k\widetilde{\bm{\Theta}},
\end{align}
where $\bm{\Theta}_{k}$ is the holographic pattern generated in the $k$-th iteration, $\delta\ll 1$ is the step length in the $k$-th iteration. Further, the diagonal matrix $\widetilde{\bm{\Theta}}=\diag\{\widetilde{\theta}_1,\dots,\widetilde{\theta}_N\}\in\mathbb{C}^{N\times N}$ represents normalized change of the holographic pattern, where each diagonal element satisfies $\widetilde{\theta}_n\in[-1,1]$. Therefore, the holographic beamformer in the $k$-th iteration is rewritten as
\begin{align}
	\label{B_k}
	\bm{B}_k&=k((e^{j\tau}\bm{\Theta}_k)^{-1}-\bm{G})^{-1}\widetilde{\bm{F}}^{ref}\\
	&=k((e^{j\tau}\bm{\Theta}_{k-1})^{-1}-\bm{G}-e^{-j\tau}\delta_k\widetilde{\bm{\Theta}})^{-1}\widetilde{\bm{F}}^{ref}
\end{align} 

Based on (\ref{delta_Theta}), by replacing the holographic pattern $\bm{\Theta}$ involved in the constraints of problem (\ref{opt_v3_problem_HBF}) with $(\bm{\Theta}_{k-1}^{-1}-\delta_k\widetilde{\bm{\Theta}})^{-1}$, we acquire the subproblem for the $k$-th iteration,~i.e.,
{
	\setlength{\abovedisplayskip}{-6pt}
	\setlength{\belowdisplayskip}{3pt}
\begin{subequations}\label{opt_v4_problem_HBF}
	\begin{align}
		\label{obj_v4_HBF}
		&\max_{\widetilde{\bm{\Theta}},\tau}  \tau  \hfill\\
		\label{cons_v4_theta_change}
		s.t.~&\widetilde{\theta}_n\in[-1,1],\hfill\\
		\label{cons_v4_theta_range}
		&(\theta_n)_{k-1}^{-1}-\delta_k\widetilde{\theta}_n\in[\theta_{max}^{-1},\theta_{min}^{-1}],\hfill\\
		&(\ref{cons_v2_aux_HBF}), (\ref{cons_max_sensing_HBF})-(\ref{cons_min_rate_HBF}),\hfill \notag
	\end{align}
\end{subequations}}
where constraint (\ref{cons_v4_theta_range}) indicates that the polarizability magnitude of each RHS element should fall within the range specified in (\ref{cons_pol_range_HBF}). 

Since $\delta_k$ takes a small value, problem (\ref{opt_v4_problem_HBF}) can be simplified through applying the Neumann series to the holographic beamformer $\bm{B}_k$, as shown below.
\begin{lemma}
	\label{lemma_approx_HB}
	When the step length $\delta_k$ satisfies $\delta_k\ll \|\bm{S}_k\|^{-1}$, the holographic beamformer in (\ref{Holographic_BF}) can be simplified as
	\begin{align}
		\label{approx_HB}
		\bm{B}_k\approx ke^{j\tau}\bm{S}_k\widetilde{\bm{F}}^{ref}+ke^{j\tau}\delta_k\bm{S}_k\widetilde{\bm{\Theta}}\bm{S}_k\widetilde{\bm{F}}^{ref},
	\end{align}
	where $\bm{S}_k=e^{-j\tau}((e^{j\tau}\bm{\Theta}_{k-1})^{-1}-\bm{G})^{-1}$ is a constant matrix, and $\|\cdot\|$ indicates the spectral norm of a matrix. 
\end{lemma}
\begin{proof}
	See Appendix~\ref{app_approx_HB}.
\end{proof}

{It should be noted that the choice of operator within the Neumann series can influence the ultimate performance. There are two possible operators within the Neumann series to be selected. The first operator, as proposed above, is $ \bm{S}_k\Delta\widetilde{\bm{\Theta}} $, which corresponds to the case where the inverse of the polarizability, i.e., $\widetilde{\bm{\Theta}}_{k}$, is regarded as optimization variables and small perturbations are applied to them. The second operator is $ \bm{S}_k\bm{\Theta}_k^{-1}\bm{\Theta}_k^{-1}\Delta{\bm{\Theta}} $, which corresponds to the case where the polarizability ${\bm{\Theta}}_{k}$ are regarded as the optimization variables. Through sensitivity analysis, we can show that although the sensitivity of the objective function with respect to the optimization variables $\bm{\Theta}_{k-1}$ (i.e., when the second operator is selected) is larger than that with respect to $\widetilde{\bm{\Theta}}_{k-1}$ (i.e., when the first operator is selected), the maximum allowed step length $\|\Delta{\bm{\Theta}}\|^{max}$ is smaller than $\|\Delta\widetilde{\bm{\Theta}}\|^{max}$, indicating that the convergence speed corresponding to the two operators are comparable. However, the first operator can achieve better performance than the second one. This is mainly because unlike the first operator, when the second operator is selected, the perturbations should be quite small. Such small perturbations could result in significant numerical errors during the operations of the proposed algorithm due to finite word length effect, which can deteriorate the algorithm's performance.}

As indicated in Lemma~\ref{lemma_approx_HB}, the holographic beamformer can be approximated by a linear function of the normalized change $\widetilde{\bm{\Theta}}$ of the holographic pattern, based on which the holographic beamforming problem can be simplified, as shown in the following theorem.
% which can be utilized to simplify the holographic beamforming problem in (\ref{opt_v4_problem_HBF}), as shown in the following theorem.
%\vspace{-.2cm}
\begin{theorem}
	\label{theorem_equivalent_HBF}
	Based on Lemma~(\ref{lemma_approx_HB}), the holographic beamforming subproblem in (\ref{opt_v4_problem_HBF}) is equivalently reformulated as
	{
		\setlength{\abovedisplayskip}{-6pt}
		\setlength{\belowdisplayskip}{3pt}
	\begin{subequations}\label{opt_v5_problem_HBF}
	\begin{align}
			\label{obj_v5_HBF}
			&\max_{\bm{\Xi},\tau}  \tau  \hfill\\
			s.t.~&-1\le\Tr(\bm{C}_n\Xi)\le 1,n=1,\dots,N,\hfill\\
			&\theta_{max}^{-1}\le(\theta_n)_{k-1}^{-1}-\delta_k\Tr(\bm{C}_n\Xi)\le \theta_{min}^{-1},n=1,\dots,N,\hfill\\
			&\Tr(\hat{\bm{U}}_{k,l}\Xi)+\left(2\Re(e_{k,l})+|\rho_l|^2(b_{k,l}-b_k-\sigma^2)\right)\ge \tau,\notag\\
			&\quad\quad\quad\quad\quad\quad\quad\quad\quad\quad\quad\quad\quad\quad l=1,\dots,L,\hfill\\
			&(\Tr(\widetilde{\bm{U}}_{k,w}\Xi)+b_{k,w})\le G_w^{th},w=1,\dots,W,\hfill\\
			&(\Tr(\widetilde{\bm{U}}_{k,d}\Xi)+ b_{k,d})-\gamma_d^l(\Tr(\widetilde{\bm{U}}_{k,1}\Xi)+ b_{k,1})\ge 0,\hfill\notag\\
			&\quad\quad\quad\quad\quad\quad\quad\quad\quad\quad\quad\quad\quad\quad d=2,\dots,D,\\
			&(\Tr(\widetilde{\bm{U}}_{k,d}\Xi)+ b_{k,d})-\gamma_d^u(\Tr(\widetilde{\bm{U}}_{k,1}\Xi)+ b_{k,1})\le 0,\notag\\
			&\quad\quad\quad\quad\quad\quad\quad\quad\quad\quad\quad\quad\quad\quad d=2,\dots,D,\hfill\\
			&(\Tr(\widetilde{\bm{U}}_{k,d}\Xi)+b_{k,d})\ge G_d^{th},d=1,\dots,D,\hfill\\
			&\frac{2^{R_l^{th}}}{2^{R_l^{th}}-1}(\Tr(\widetilde{\bm{U}}_{k,l}\Xi)+b_{k,l})-(\Tr(\widetilde{\bm{U}}_{k}\Xi)+b_{k})\ge \sigma^2,\notag\\
			&\quad\quad\quad\quad\quad\quad\quad\quad\quad\quad\quad\quad\quad\quad l=1,\dots,L,\hfill\\
			&\Xi=
			\begin{bmatrix}
			\newvec(\widetilde{\bm{\Theta}})\newvec(\widetilde{\bm{\Theta}})^{\mathsf{T}}	 & t\newvec(\widetilde{\bm{\Theta}}) \\
				t\newvec(\widetilde{\bm{\Theta}})^T & t^2
			\end{bmatrix}\\
			&t^2=1,\hfill
			\end{align}
\end{subequations}}
where the definitions of constant matrices $\bm{C}_n$, $\hat{\bm{U}}_{k,l}$, $\widetilde{\bm{U}}_{k,d}$, $\widetilde{\bm{U}}_{k,l}$, and $\widetilde{\bm{U}}_{k}$, and the definitions of constant scalar $e_{k,l}$, $b_{k,l}$, $b_k$ can be found in Appendix~\ref{appendix_equivalent_HBF}. 
\end{theorem}
%\begin{proof}
%	See Appendix~\ref{appendix_equivalent_HBF}.
%\end{proof}

The problem in (\ref{opt_v5_problem_HBF}) is a quadratic program, which can be effectively solved using the SDR method~\cite{Luo_SDR_2010}. The overall algorithm is summarized in Algorithm~\ref{algorithm_holographic_BF}.

%\begin{theorem}
%	By substituting (\ref{approx_HB}) into (\ref{opt_v3_problem_HBF}), the holographic beamforming problem can be reformulated as a QCQP problem, as shown below,
%	\begin{subequations}\label{opt_v4_problem_HBF}
%		\begin{align}
%			\label{obj_v4_HBF}
%			&\max_{\widetilde{\bm{\Theta}},\tau}  \tau  \hfill\\
%			s.t.~&2\Re(\bm{d}_{k,l}-|\rho_l|^2\bm{c}_k+|\rho_l|^2\bm{c}_{k,l})^T\vec{\widetilde{\bm{\Theta}}}+\vec{\widetilde{\bm{\Theta}}}^T(|\rho_l|^2\bm{U}_{k,l}-|\rho_l|^2\bm{U}_k)\vec{\widetilde{\bm{\Theta}}}+(2\Re(e_{k,l})-|\rho_l|^2b_k+|\rho_l|^2b_{k,l}-|\rho_l|^2\sigma^2)\ge \tau,l=1,\dots,L\hfill\\
%			&(2\Re(\bm{c}_{k,d})^T\vec{\widetilde{\bm{\Theta}}}+\vec{\widetilde{\bm{\Theta}}}^T\bm{U}_{k,d}\vec{\widetilde{\bm{\Theta}}}+b_{k,d})-\gamma_d^l(2\Re(\bm{c}_{k,1})^T\vec{\widetilde{\bm{\Theta}}}+\vec{\widetilde{\bm{\Theta}}}^T\bm{U}_{k,1}\vec{\widetilde{\bm{\Theta}}}+b_{k,1})\ge 0,d=2,\dots,D\hfill\\
%			&(2\Re(\bm{c}_{k,d})^T\vec{\widetilde{\bm{\Theta}}}+\vec{\widetilde{\bm{\Theta}}}^T\bm{U}_{k,d}\vec{\widetilde{\bm{\Theta}}}+b_{k,d})-\gamma_d^u(2\Re(\bm{c}_{k,1})^T\vec{\widetilde{\bm{\Theta}}}+\vec{\widetilde{\bm{\Theta}}}^T\bm{U}_{k,1}\vec{\widetilde{\bm{\Theta}}}+b_{k,1})\le 0,d=2,\dots,D\hfill	\end{align}
%	\end{subequations}
%\end{theorem}

\begin{algorithm}[!tpb]
	\caption{Holographic beamforming design}
	\label{algorithm_holographic_BF} % Place label right after caption
	\begin{algorithmic}[1]
		\REQUIRE Digital beamformer $\bm{V}$
%		\STATE Initialize holographic pattern $\bm{\Theta}$
		\REPEAT
		\STATE Given $\bm{\Theta}$, compute $\{\rho_l^{opt}\}_l$ for all $L$ users according to (\ref{rho_opt_HBF})
		\STATE Set $k=1$
		\REPEAT
		\STATE Set maximum step length $\delta_k$
		\STATE Given $\{\rho_l\}_l$, optimize $\widetilde{\bm{\Theta}}$ by solving problem (\ref{opt_v5_problem_HBF}) through the SDR method
		\STATE Update $\bm{\Theta}_k$ according to (\ref{delta_Theta})
		\STATE Set $k=k+1$
		\UNTIL{Convergence}
		\UNTIL{Convergence}
		\ENSURE{The optimal holographic pattern $\bm{\Theta}^{*}$}
	\end{algorithmic}
\end{algorithm}

\subsection{Overall Algorithm Description}
Based on the algorithms presented in the previous two subsections, we design a joint mutual coupling aware sum-rate optimization algorithm to solve problem (\ref{opt_problem}) iteratively. Specifically, in each iteration, the digital beamformer is first optimized by solving problem (\ref{opt_problem_DBF}), given the holographic pattern $\bm{\Theta}$ at the RHS (i.e., the amplitude of polarizability of the equivalent magnetic dipoles). Then, given the digital beamformer $\bm{V}$, the holographic pattern is optimized through Algorithm~\ref{algorithm_holographic_BF}. The iterations are performed until the value difference of the minimum data rate among all users between two adjacent iterations is less than a predefined threshold. The overall algorithm is summarized in Algorithm~\ref{algorithm_all}.

\begin{algorithm}[!tpb]
	\caption{Mutual coupling aware joint optimization algorithm design}
	\label{algorithm_all}
	\begin{algorithmic}[1] % [1] enables line numbers
		%		\REQUIRE 
		\STATE Initialize digital beamformer $\bm{V}$ and holographic pattern $\bm{\Theta}$
		\STATE Set $o=1$		
		\REPEAT
		\STATE Given~$\bm{\Theta}^{(o-1)}$, derive~$\bm{V}^{(o)}$ by solving problem~(\ref{opt_problem_DBF_v6});
		\STATE Given $\bm{V}^{(o)}$, calculate $\bm{\Theta}^{(o)}$ by solving problem (\ref{opt_problem_HBF}) using Algorithm~\ref{algorithm_holographic_BF};
		\STATE Update $o=o+1$
		\UNTIL{Convergence}		
		\ENSURE Optimal digital beamformer $\bm{V}^{opt}$ and optimal holographic pattern $\bm{\Theta}^{opt}$.
	\end{algorithmic}
\end{algorithm}
\vspace{-.4cm}
%\begin{algorithm}[!tpb]{
%		\caption{Transmit Covariance Optimization Algorithm}
%		\begin{algorithmic}[1]\label{algorithm_covariance}
%			\REQUIRE{Statistical channel information $\bm{\Omega}$, transmit SNR $\gamma$, user distance $r_U$, and non-uniform XPD distance $r_U^{th}$}
%			\IF{$r_U>r_U^{th}$}
%			\STATE{Set scalar matrix as optimal transmit covariance matrix, i.e., $\bm{Q}=\frac{1}{2M}\bm{I}_{2M}$.}
%			\ELSE
%			\STATE{Initialize $\{(\widetilde{\lambda}_s^{(V)})^*_{0},(\widetilde{\lambda}_s^{(H)})^*_{0}\}=\frac{1}{2S}$}
%			\REPEAT 
%			\STATE Set $\{(\widetilde{\lambda}_s^{(V)})^*_{o-1},(\widetilde{\lambda}_s^{(H)})^*_{o-1}\}$ as the initial solution;
%			\STATE Given a penalty parameter $\mu_o$, solve problem (\ref{opt_problem_power_v4}) using gradient descent. Denote the optimal solution by $\{(\widetilde{\lambda}_s^{(V)})^*_{o},(\widetilde{\lambda}_s^{(H)})^*_{o}\}$;
%			\STATE Let $\mu_{o+1}>\mu_o$, and $o=o+1$.
%			\UNTIL {The objective function in (\ref{obj_power_v3}) converges and the constraints (\ref{cons_semi_def_power_v3})-(\ref{cons_max_each_power_power_v3}) are satisfied.}
%			\STATE Acquire the optimal transmit covariance by substituting (\ref{optimal_unitary}) and the derived optimal power allocation into (\ref{diagnalization}).
%			\ENDIF
%			\ENSURE{Optimal transmit covariance matrix $\bm{Q}$}
%	\end{algorithmic}}
%\end{algorithm}
%\vspace{-.3cm}
\subsection{Convergence and Complexity Analysis}

\subsubsection{Convergence} In the $o$-th iteration, the sum-rate is non-decreasing after the digital beamformer $\bm{V}$ is optimized given the holographic pattern $\bm{\Theta}^{(o-1)}$, i.e.,
\begin{align}
	\label{non_decreasing_1}
	R_{sum}(\bm{V}^{(o)},\bm{\Theta}^{(o-1)})\ge R_{sum}(\bm{V}^{(o-1)},\bm{\Theta}^{(o-1)}).
\end{align}
Then, given $\bm{V}^{(o)}$, the holographic pattern is optimized through Algorithm~\ref{algorithm_holographic_BF}, such that we have
\begin{align}
	\label{non_decreasing_2}
	R_{sum}(\bm{V}^{(o)},\bm{\Theta}^{(o)})\ge R_{sum}(\bm{V}^{(o)},\bm{\Theta}^{(o-1)}).
\end{align}
Based on (\ref{non_decreasing_1}) and (\ref{non_decreasing_2}), we can obtain that
\begin{align}
	R_{sum}(\bm{V}^{(o)},\bm{\Theta}^{(o)})\ge R_{sum}(\bm{V}^{(o-1)},\bm{\Theta}^{(o-1)}),
\end{align}
which indicates that in each update step of the proposed algorithm, the objective function value (i.e., minimum data rate) is non-decreasing. Since the objective function value sequence obtained in the iteration steps is monotonic and also bounded, the overall algorithm is guaranteed to converge.

\subsubsection{Complexity} We evaluate the complexity of the proposed algorithms for each of the two subproblems individually and subsequently outline the complexity of the overall optimization algorithm.

In the digital beamforming subproblem, in each iteration, the complexity for calculating $\{\rho_l\}$ is $\mathcal{O}(L)$. Further, the complexity of SDR for solving problem (\ref{opt_problem_DBF_v6}) is $\mathcal{O}(K^{4.5}\log(1/\xi))$, where $\xi$ is the solution accuracy of the SDR algorithm. Denote the number of overall iterations by $I_D$. Then, the computation complexity for the digital beamforming algorithm is $\mathcal{O}(I_D(L+K^{4.5}\log(1/\xi)))$.

For the holographic beamforming subproblem, denote the number of outer and inner iterations by $I_H^{(out)}$ and $I_H^{(in)}$. In each inner iteration, the complexity for calculating auxiliary variable $\{\rho_l\}$ is $\mathcal{O}(L)$. Then, problem (\ref{opt_v5_problem_HBF}) is solved using the SDR technique, with complexity of $\mathcal{O}(N^{4.5}\log(1/\xi))$. Therefore, the complexity for the overall holographic beamforming algorithm is $\mathcal{O}(I_H^{(out)}I_H^{(in)}(L+N^{4.5}\log(1/\xi)))$.

To acquire the complexity for overall holographic ISAC algorithm in Algorithm~\ref{algorithm_all}, we first derive the number of iterations for optimizing the digital beamformer and holographic pattern alternatively. Note that the sum rate of the $L$ mobile users is upper bounded by $L\log_2(1+\frac{P}{L\sigma^2})$. Therefore, the minimum data rate among these users cannot exceed $\log_2(1+\frac{P}{L\sigma^2})$. During the operation of the holographic ISAC algorithm, the change of the objective function in each iteration does not exceed the termination threshold $\epsilon$. Therefore, the number of overall iterations does not exceed $\log_2(1+\frac{P}{L\sigma^2})/\epsilon$. Consequently, the computation complexity for Algorithm~\ref{algorithm_all} is $\mathcal{O}((\log_2(1+\frac{P}{L\sigma^2})/\epsilon)(I_D(L+K^{4.5}\log(1/\xi))+I_H^{(out)}I_H^{(in)}(L+N^{4.5}\log(1/\xi))))=\mathcal{O}((1/L)(I_DK^{4.5}+I_H^{(out)}I_H^{(in)}N^{4.5})\log(1/\xi)+I_D+I_H^{(out)}I_H^{(in)})$.

\begin{figure}[!t]
	\centering
	\includegraphics[width=0.3\textwidth]{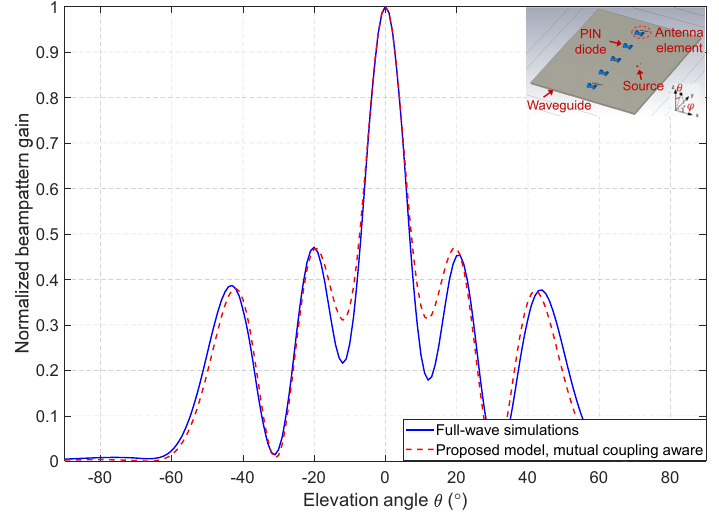}
	%		\vspace{1.5mm}
	\caption{Validation of the accuracy of the proposed mutual coupling aware model through full-wave simulations. $5$ practical cELC-based RHS elements are etched on a parallel-plate waveguide, and the beampattern is plotted given $\varphi=90^\circ$.}
	\vspace{-5mm}
	\label{radiation_pattern_cst}
\end{figure}

{
\begin{figure}[!t]
	\centering
	\includegraphics[width=0.39\textwidth]{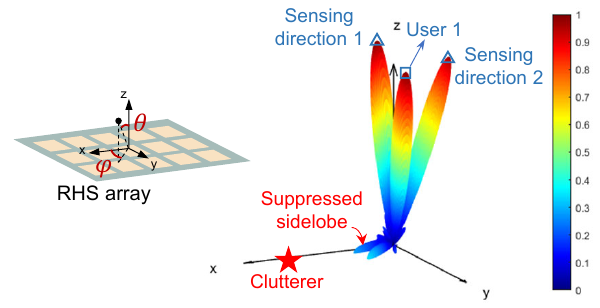}
	%		\vspace{1.5mm}
	\caption{Beampattern of an RHS array obtained using the proposed mutual coupling aware holographic beamforming algorithm. $L=1$ mobile user and $D=2$ sensing directions are considered. The number of RHS elements is set as $N=20^2$.}
		\vspace{-5mm}
	\label{fig_beampattern}
\end{figure}}

\begin{figure}[!t]
	\centering
	\includegraphics[width=0.39\textwidth]{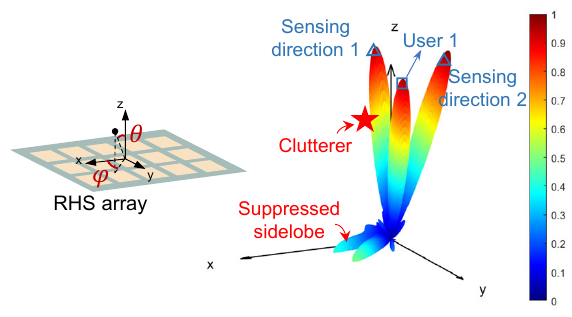}
	%		\vspace{1.5mm}
	\caption{Beampattern of an RHS array obtained using the proposed mutual coupling aware holographic beamforming algorithm. The direction of clutterer is $(20^\circ,14^\circ)$. The number of RHS elements is set as $N=20^2$.}
		\vspace{-5mm}
	\label{fig_beampattern_misaligned_clutter}
\end{figure}
%\subsubsection{Influence of discrete polarizability control}:
\vspace{-.2cm}
\section{Simulation Results}
\label{sec_simulation}
\vspace{-.2cm}
In this section, we evaluate the performance of the proposed mutual coupling aware holographic ISAC algorithm. For simplicity, we consider line-of-sight propagations between the BS and each user. Major simulation parameters are set up based on and existing works~\cite{Yoo_Sub_2023,Mancera_analytical,Zeng_coverage_2021,Zeng_both_2021,Haobo_HISAC_2022}. Specifically, the center frequency is set to $f=30$~GHz\footnote{Compared with the sub-6G frequency band, the mutual coupling effect is more significant in the mmWave band~\cite{Balanis_antenna_the}. Therefore, to demonstrate the necessity of applying the proposed mutual coupling aware algorithm, we adopt a typical mmWave working frequency here, i.e., $30$~GHz.}, with the corresponding wavelength given by $\lambda=1$~cm. Transmit power of the BS is set as $P=43$~dBm, and noise variance is $\sigma^2=-96$~dBm. Element spacing of the RHS is $0.23\lambda$. Further, the feeds of the RHS are deployed on one side of the surface aligned with the $z$-axis. The separation of the feeds are also $0.23\lambda$. We adopt the reference wave model given in~\cite{Deng_RHS_multi_user_2022}, i.e., the reference wave at the $n$-th RHS element generated by the $l$-th feed is given by $\widetilde{\bm{f}}_{l}^{ref}(n)=\exp(-j\bm{k}_s\cdot \bm{r}_{l,n})$. Here, $\bm{k}_s$ is the propagation vector of the reference wave within the RHS, which is set as $\bm{k}_s=\sqrt{3}\frac{2\pi}{\lambda}$. Further, $\bm{r}_{l,n}$ is the distance vector from the $l$-th feed to the $n$-th RHS element. The Green's function $\bm{G}$ is modeled based on the theoretical model in~\cite{Mancera_analytical}, which considers the mutual coupling among elements via guided and radiated fields. We assume that each UE is equipped with $J=1$ antenna elements, with an element spacing of half wavelength. $L = 1$ communication users and $D = 2$ sensing targets are considered in Fig.~\ref{fig_beampattern}~Fig.~\ref{fig_beampattern_all_v9} while $L = 4$ communication users and $D = 2$ sensing targets are considered in Fig.~\ref{fig_beampattern_vs_iteration_rounds}~Fig.~\ref{fig_convergence_rate}. $W=1$ environmental clutterer is considered. The upper and lower bound for the beampattern gains are set as $\gamma_d^l=0.9$ and $\gamma_d^u=1.1$ for $d=1,\dots,D$. The minimum required data rate is set as $1$~bps/Hz for all communication users.

{To demonstrate that the proposed model can capture the mutual coupling effects, we compare the radiation pattern acquired through the proposed model against that through full-wave numerical simulations, as shown in Fig.~\ref{radiation_pattern_cst}. Here, we plot the beampattern of an antenna array given the azimuth angle $\varphi=90^\circ$. In the considered antenna array, $5$ practical cELC-based RHS elements (with the same structure as that in~\cite{Deng_VTM_2023}) are etched on a parallel-plate waveguide. Further, an electric dipole is deployed within the waveguide as excitation source, where the dipole is placed vertical to the waveguide surface. The frequency is set as $27.1$~GHz. Each RHS elements are equipped with two PIN diodes to control its radiation amplitude (i.e., the polarizability magnitude in the context of coupled dipole models). For simplicity of discussions, we set all RHS elements to the ``ON" state, i.e., the RHS elements are resonant and can radiate power into free space. However, the radiating fields of the RHS elements can also propagate to the other RHS elements, resulting in mutual coupling. In Fig.~\ref{radiation_pattern_cst}, the full-wave simulations are accomplished using CST Microwave Suite. From Fig.~\ref{radiation_pattern_cst}, we can see that the proposed mutual coupling aware model can achieve similar radiation pattern to that acquired through full-wave simulations, which, thus demonstrates that the proposed model can capture the mutual-coupling effects.}

Fig.~\ref{fig_beampattern} illustrates the beampattern generated by the RHS using the proposed mutual coupling aware holographic ISAC method. $L=1$ communication users and $D=2$ angles of interest for sensing are considered, with their directions $(\theta,\phi)$ given by $(20^\circ,80^\circ)$, $(20^\circ,260^\circ)$, and $(20^\circ,170^\circ)$, respectively\footnote{When the angle difference between entities decreases, both communication and sensing performance can degrade.}. We assume there is $W=1$ environmental clutterer, with its direction given by $(90^\circ,14^\circ)$. Further, the maximum and minimum polarizability magnitudes are set as $\theta_{max}=3.01\times10^{-6}$ and $\theta_{min}=3.01\times10^{-8}$, respectively. A comparison with Fig.~\ref{Superposability}, which shows the radiation pattern of the RHS under the existing holographic ISAC method designed without considering mutual coupling~\cite{Haobo_HISAC_2022}, reveals that the proposed algorithm adapts better to ISAC systems in clutter environments. This is because the sidelobe level reduces from $-1.88$~dB to $-4.57$~dB, indicating lower echo signal power from the environmental clutterers and less interference to sensing procedures. {The beamwidth of the beams in directions $(\theta,\varphi)=(20^\circ,80^\circ), (20^\circ,170^\circ), (20^\circ,260^\circ)$ are $15.32^\circ$, $12.45^\circ$, and $12.42^\circ$, respectively, which can satisfy communication and sensing requirements.} 

{Based on the simulation settings in Fig.~\ref{fig_beampattern}, we further change the clutterer angle to $\theta=20^{\circ}$, and present the corresponding radiation pattern of the RHS in Fig.~\ref{fig_beampattern_misaligned_clutter}. As observed, the sidelobe level is $-3.06$~dB, and thus the sidelobe reduction effect of the proposed algorithm becomes relatively smaller compared to that in Fig.~\ref{fig_beampattern}. This is because the direction of the unexpected sidelobe, which arises from ignoring mutual coupling in holographic beamforming, is mainly determined by the locations of the communication users and sensing targets. Therefore, given the locations of the communication users and sensing targets, the direction of the sidelobe cannot change freely and cannot remain aligned with the clutterer when the location of the clutterer changes. Consequently, when the clutterer moves to a new direction $\theta=20^{\circ}$, the clutterer is no longer aligned with the sidelobe, as shown in Fig.~\ref{fig_beampattern_misaligned_clutter}. Consequently, the sidelobe has a reduced impact on sensing procedures, and thus the sidelobe reduction achieved by the proposed algorithm becomes relatively smaller.}

%Compared withFig.~\ref{fig_beampattern_misaligned_clutter} demonstrates
%According to Fig.~\ref{fig_beampattern_misaligned_clutter}, when the clutter angle changes to $\theta=20^{\circ}$, the sidelobe reduction is less significant.

\begin{figure}[!tpb]
	\centering
	
	\subfigure[Beampattern of an RHS array without mutual coupling awareness]{
		\begin{minipage}[b]{0.4\textwidth}
			\centering
			%			\vspace{-0.2cm}
			\includegraphics[width=.95\textwidth]{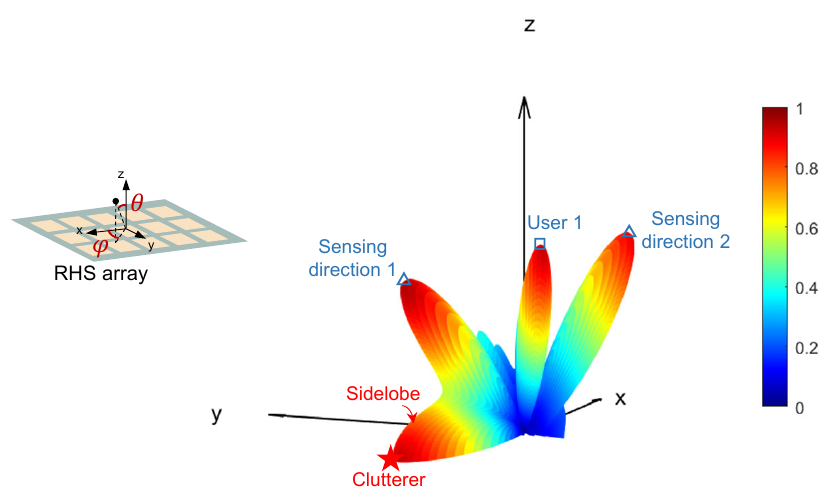}
			%			\vspace{-0.4cm}
			%		\caption{Phase response vs. frequency}
			%		\vspace{-0.3cm}
			\label{Superposability_v3}
	\end{minipage}}

	\subfigure[Beampattern of an RHS array using the proposed mutual coupling aware algorithm]{
		\begin{minipage}[b]{0.4\textwidth}
			\centering
						\vspace{-0.4cm}
			\includegraphics[width=.95\textwidth]{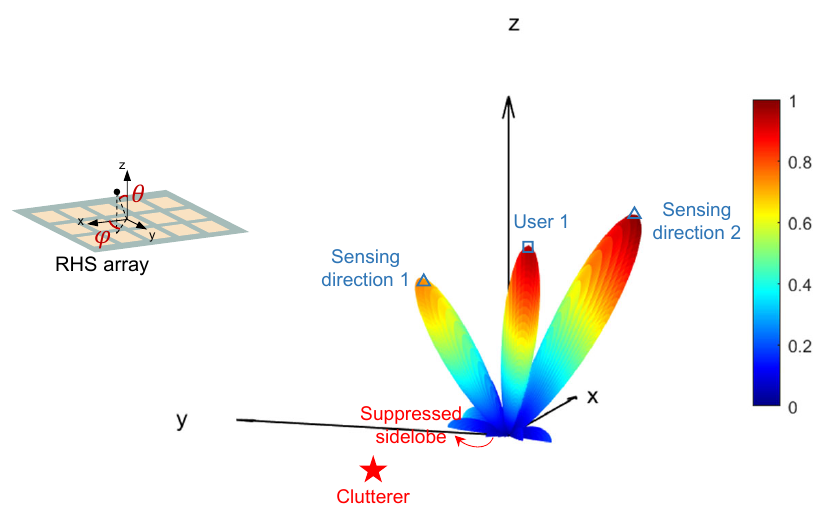}
						\vspace{-0.2cm}
			%		\caption{Amplitude vs. frequency}
			%		\vspace{-0.cm}
			\label{fig_beampattern_v4}
	\end{minipage}}
	%	\vspace{-0.1cm}
	\caption{Evaluation of sidelobe suppresion performance of the proposed mutual coupling aware method when the mobile users, targets and clutterers are spatially close. One mobile user, two sensing directions, and one environmental clutterer are considered, with their directions given by $(\theta,\phi)=(50^\circ,195^\circ)$, $(50^\circ,165^\circ)$, $(50^\circ,225^\circ)$, and $(90^\circ,169^\circ)$, respectively. There are $N=20^2$ RHS elements.}
	\label{fig_beampattern_all_v8}
		\vspace{-0.5cm}
\end{figure}

%\begin{figure}[!t]
%	\centering
%	\includegraphics[width=0.45\textwidth]{fig_first_round_3D_coupled_no_opt_close_clutter.pdf}
%	%		\vspace{1.5mm}
%	\caption{Beampattern of an RHS array when mutual coupling effects are not considered in holographic beamforming~[19]. The One mobile users and two sensing directions are considered, with one clutter in the environment. The mobile users, targets and clutters are spatially close, with their elevation angles given by $\phi=195^\circ, 165^\circ, 225^\circ, 169^\circ$, respectively. The number of RHS elements is set as $N=20^2$.}
%	\vspace{-3mm}
%	\label{Superposability_v3}
%\end{figure}
%
%\begin{figure}[!t]
%	\centering
%	\includegraphics[width=0.45\textwidth]{fig_first_round_3D_coupled_opt_close_clutter.pdf}
%	%		\vspace{1.5mm}
%	\caption{Beampattern of an RHS array obtained using the proposed mutual coupling aware holographic ISAC algorithm. $L=1$ mobile user and $D=2$ sensing directions are considered. The mobile users, targets and clutters are spatially close, with their elevation angles given by $\phi=195^\circ, 165^\circ, 225^\circ, 169^\circ$, respectively. The number of RHS elements is set as $N=20^2$.}
%	%	\vspace{-3mm}
%	\label{fig_beampattern_v4}
%\end{figure}

{
In Fig.~\ref{fig_beampattern_all_v8}, we extend discussions to the case where the mobile users, targets and clutterers are spatially close, and show the radiation pattern of the RHS when the mutual coupling effect is not considered and is taken into account, respectively. The azimuth angles of the mobile user, the two sensing targets, and the clutterer are given by $\phi=195^\circ, 165^\circ, 225^\circ, 169^\circ$, respectively. The elevation angles for the mobile user and the sensing targets are set as $\theta=50^\circ$ while that for the clutterer is $\theta=90^\circ$. By comparing Fig.~\ref{Superposability_v3} and Fig.~\ref{fig_beampattern_v4}, we can see that when the mutual coupling effect is considered during holographic beamforming, the sidelobe of the RHS can be effectively suppressed, from $-0.08$~dB to $-17.24$~dB, which thus demonstrates the effectiveness of the proposed algorithm. Further, compared with Fig.~\ref{fig_beampattern}, we can see that when the angle separation between the clutterer and the sensing target/communication user becomes smaller, the sidelobe suppression effect becomes more evident. This is mainly because a smaller separation results in stronger power leakage from the main lobe into the sidelobe, thereby increasing the sidelobe level. Additionally, the array response in the directions of the main lobe and that for the sidelobe become more correlated, which further amplifies the sidelobe level. Consequently, the proposed mutual coupling aware algorithm achieves more effective sidelobe suppression in such scenarios.}

\begin{figure}[!tpb]
	\centering
	
	\subfigure[Beampattern of an RHS array without mutual coupling awareness]{
		\begin{minipage}[b]{0.42\textwidth}
			\centering
			%			\vspace{-0.2cm}
			\includegraphics[width=.95\textwidth]{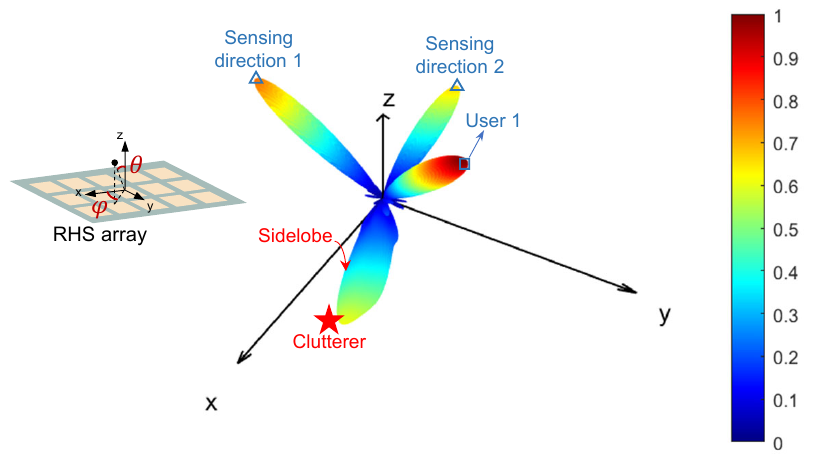}
			%			\vspace{-0.4cm}
			%		\caption{Phase response vs. frequency}
			%		\vspace{-0.3cm}
			\label{Superposability_v4}
	\end{minipage}}
	
	\subfigure[Beampattern of an RHS array using the proposed mutual coupling aware algorithm]{
		\begin{minipage}[b]{0.41\textwidth}
			\centering
						\vspace{-0.2cm}
			\includegraphics[width=.95\textwidth]{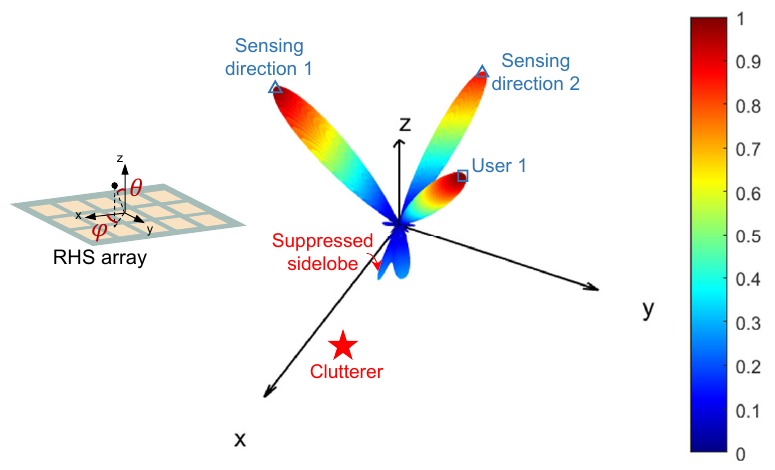}
						\vspace{-0.2cm}
			%		\caption{Amplitude vs. frequency}
			%		\vspace{-0.cm}
			\label{fig_beampattern_v5}
	\end{minipage}}
	%	\vspace{-0.1cm}
	\caption{Evaluation of sidelobe suppresion performance of the proposed mutual coupling aware method when the elevation angles $\theta$ of the mobile users, targets and clutterers are different. One mobile user, two sensing directions, and one environmental clutterer are considered, with their directions given by $(\theta,\phi)=(20^\circ,80^\circ)$, $(30^\circ,170^\circ)$, $(40^\circ,260^\circ)$, and $(90^\circ,14^\circ)$, respectively. The number of RHS elements is set as $N=20^2$.}
	\label{fig_beampattern_all_v9}
	\vspace{-0.5cm}
\end{figure}

%\begin{figure}[!t]
%	\centering
%	\includegraphics[width=0.45\textwidth]{first_round_3D_beam_sidelobe_coupled_no_opt_dif_theta_20&30&40.pdf}
%	%		\vspace{1.5mm}
%	\caption{Beampattern of an RHS array when mutual coupling effects are not considered in holographic beamforming~[19]. One mobile users and two sensing directions are considered, with one clutter in the environment. The elevation angles of the mobile user, the two sensing targets, and the clutter are different, which are set as $\theta=20^\circ, 30^\circ, 40^\circ, 90^\circ$, respectively. The number of RHS elements is set as $N=20^2$.}
%	\vspace{-3mm}
%	\label{Superposability_v4}
%\end{figure}
%
%\begin{figure}[!t]
%	\centering
%	\includegraphics[width=0.45\textwidth]{first_round_3D_beam_sidelobe_coupled_opt_dif_theta_20&30&40.pdf}
%	%		\vspace{1.5mm}
%	\caption{Beampattern of an RHS array obtained using the proposed mutual coupling aware holographic ISAC algorithm. $L=1$ mobile user and $D=2$ sensing directions are considered. The elevation angles of the mobile user, the two sensing targets, and the clutter are different, which are set as $\theta=20^\circ, 30^\circ, 40^\circ, 90^\circ$, respectively. The number of RHS elements is set as $N=20^2$.}
%	%	\vspace{-3mm}
%	\label{fig_beampattern_v5}
%\end{figure}

{
In addition, we also evaluate the sidelobe interference suppression performance of the proposed method when the elevation angles of the mobile users and sensing targets are different, as shown in Fig.~\ref{Superposability_v4} and Fig.~\ref{fig_beampattern_v5}. Here, the directions of the mobile user, the two sensing targets, and the clutterer are set as $(\theta,\phi)=(20^\circ,80^\circ)$, $(30^\circ,170^\circ)$, $(40^\circ,260^\circ)$, and $(90^\circ,14^\circ)$, respectively. By comparing Fig.~\ref{fig_beampattern_v5} against Fig.~\ref{Superposability_v4}, we can see that by adopting the proposed mutual-coupling aware holographic beamforming method, the generated sidelobe can be effectively suppressed, which indicates the robustness of the proposed method.}

%These results demonstrate the proposed algorithm's capability to mitigate beam distortion caused by mutual coupling effects, enabling precise multi-beam radiation.

%can significantly enhance system performances. Specifically, it accurately directs the RHS beams toward the communication users and targeted sensing direction while effectively suppressing sidelobes. The beam deviations are reduced to $0.6^\circ$, $2.1^\circ$, $3.9^\circ$ for mobile user $1$, mobile user $2$, and sensing direction $1$, respectively, while the sidelobe level is suppressed to $-2.7$~dB. These results demonstrate the proposed algorithm's capability to mitigate beam distortion caused by mutual coupling effects, enabling precise multi-beam radiation.

%Fig.~\ref{multi_beam_optimized} depicts the beam pattern generated by the RHS through the proposed mutual coupling aware holographic beamforming method in a two-user case. By comparing Fig.~\ref{multi_beam_optimized} with Fig.~\ref{Superposability}, which records the radiation pattern of the RHS when applying existing superposability-based multi-user holographic beamforming method tailored for an ideal RHS without mutual coupling effect~\cite{Deng_HDMA_2022}, we can find that the proposed mutual coupling aware algorithm can improve the system performance by accurately directing the beam of the RHS towards users, while suppressing sidelobes. This thus demonstrates that the proposed algorithm can well address the beam distortion issue brought by the mutual-coupling effects and support accurate multi-beam radiations.

\begin{figure}[!tpb]
	\centering
	\subfigure[Beampattern gain in the direction of the clutterer vs. iteration rounds]{
			\begin{minipage}[b]{0.43\textwidth}
					\centering
		%			\vspace{-0.2cm}
					\includegraphics[width=.85\textwidth]{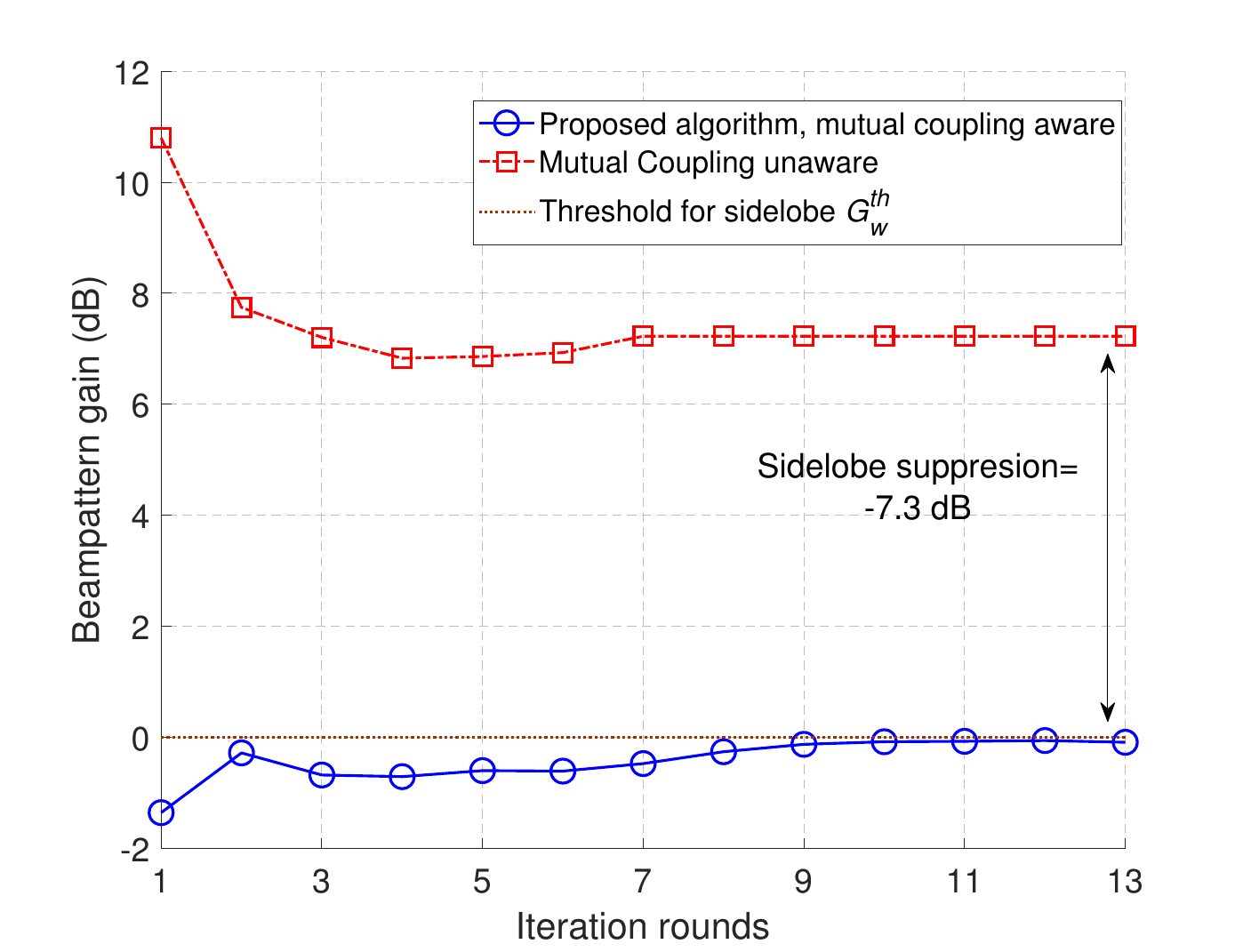}
					%			\vspace{-0.4cm}
					%		\caption{Amplitude vs. frequency}
							\vspace{-0.2cm}
					\label{Beampattern_gain_vs_iteration_rounds_sidelobe}
			\end{minipage}}
		
	\subfigure[Beampattern gain in the direction of sensing targets vs. iteration rounds]{
			\begin{minipage}[b]{0.43\textwidth}
					\centering
					\vspace{-0.3cm}
					\includegraphics[width=.85\textwidth]{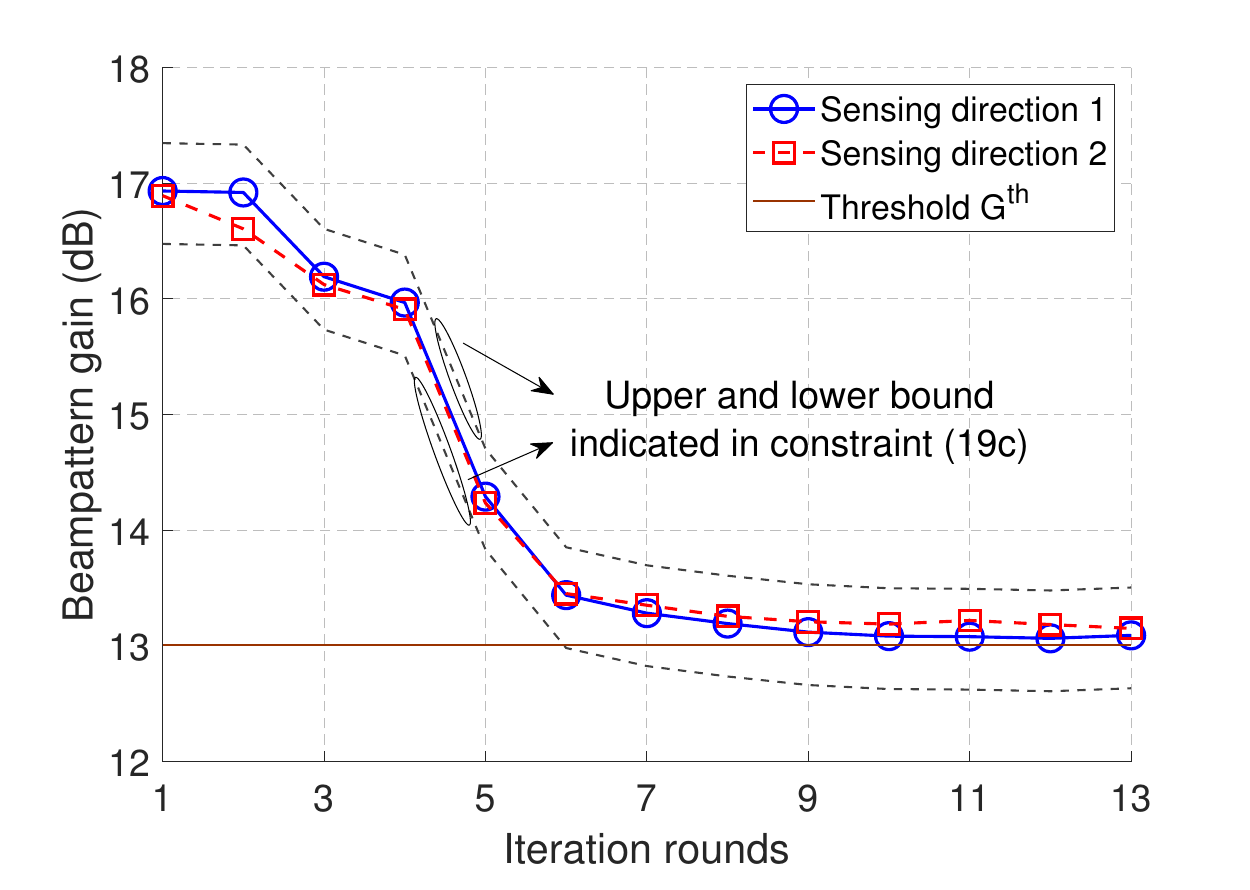}
					%			\vspace{-0.4cm}
					%		\caption{Phase response vs. frequency}
							\vspace{-0.2cm}
					\label{Beampattern_gain_vs_iteration_rounds}
			\end{minipage}}
%	\vspace{-0.1cm}
%	\caption{(a). Beampattern gain in the direction of clutter versus iteration rounds; (b). Beampattern gain in the directions of sensing targets versus iteration rounds. $L=4$ communication users and $D=2$ targeted sensing directions are considered, with the number of RF chains $K=4$, the threshold for the beampattern gain in the sensing directions $G^{th}=20$, and that in the clutter direction $G^{th}_{sb}=1$.}
	\caption{Beampattern gain versus iteration rounds, with the angle between the clutterer and the closest sensing target equal to $20.2^\circ$. The scenario considers $T = 4$ RF chains. The beampattern gain thresholds are set as $G^{\mathrm{th}} = 20$ for the sensing directions and $G^{\mathrm{th}}_{\mathrm{sb}} = 1$ for the clutterer direction.}
	\label{fig_beampattern_vs_iteration_rounds}
	\vspace{-0.6cm}
\end{figure}

Further, in Fig.~\ref{fig_beampattern_vs_iteration_rounds}, we plot the change of the beampattern gains with iteration rounds. According to Fig.~\ref{Beampattern_gain_vs_iteration_rounds_sidelobe}, we can see that when for the proposed holographic beamforming algorithm, where mutual coupling effects are considered, the sidelobe level, i.e., the beampattern gain in the direction of the clutterer, always satisfies the maximum allowed sidelobe constraint. In contrast, when the mutual coupling effects are ignored, the holographic beamforming algorithm cannot effectively suppress sidelobes as the iteration goes on, with a sidelobe levels $7.3$~dB worse than the proposed algorithm when iteration terminates. This further demonstrates the superiority of the proposed mutual coupling algorithm over existing mutual coupling unaware ones. 

{From Fig.~\ref{Beampattern_gain_vs_iteration_rounds}, we can see that as iterations continue, the beampattern gains in the direction of the sensing targets tend to decrease. This is because we aim to maximize the communication rate of the users while guaranteeing the beampattern gains in the directions of sensing target above predetermined thresholds. Note that there is a tradeoff between the communication and sensing performance in the ISAC system. Therefore, as the iteration proceeds, enhancing the communication rate requires compromising the sensing performance—specifically, by reducing the beampattern gain in the sensing directions to approach, yet remain above, the required threshold, which can be realized by jointly designing the digital and holographic beamformers\footnote{Although the beampattern gains in the sensing directions decrease with the number of iteration rounds, the sensing performance can still be guaranteed, as shown in Fig.~\ref{Beampattern_gain_vs_iteration_rounds_sidelobe} and Fig.~\ref{Beampattern_gain_vs_iteration_rounds}.}} Further, it can be observed that the beampattern gains in the two targeted sensing angles always fall within the upper and lower bound specified in constraint (\ref{cons_fair_sensing}) and satisfy the minimum beampattern gain requirement $G^{th}$. This is because according to the proposed holographic ISAC algorithm given in Algorithm~\ref{algorithm_all}, in each iteration, the digital beamformer and holographic pattern are solved to satisfy these constraints.

\begin{figure}[!tpb]
	\centering
	\subfigure[Beampattern gain in the direction of the clutterer vs. iteration rounds]{
		\begin{minipage}[b]{0.43\textwidth}
			\centering
			%			\vspace{-0.2cm}
			\includegraphics[width=.85\textwidth]{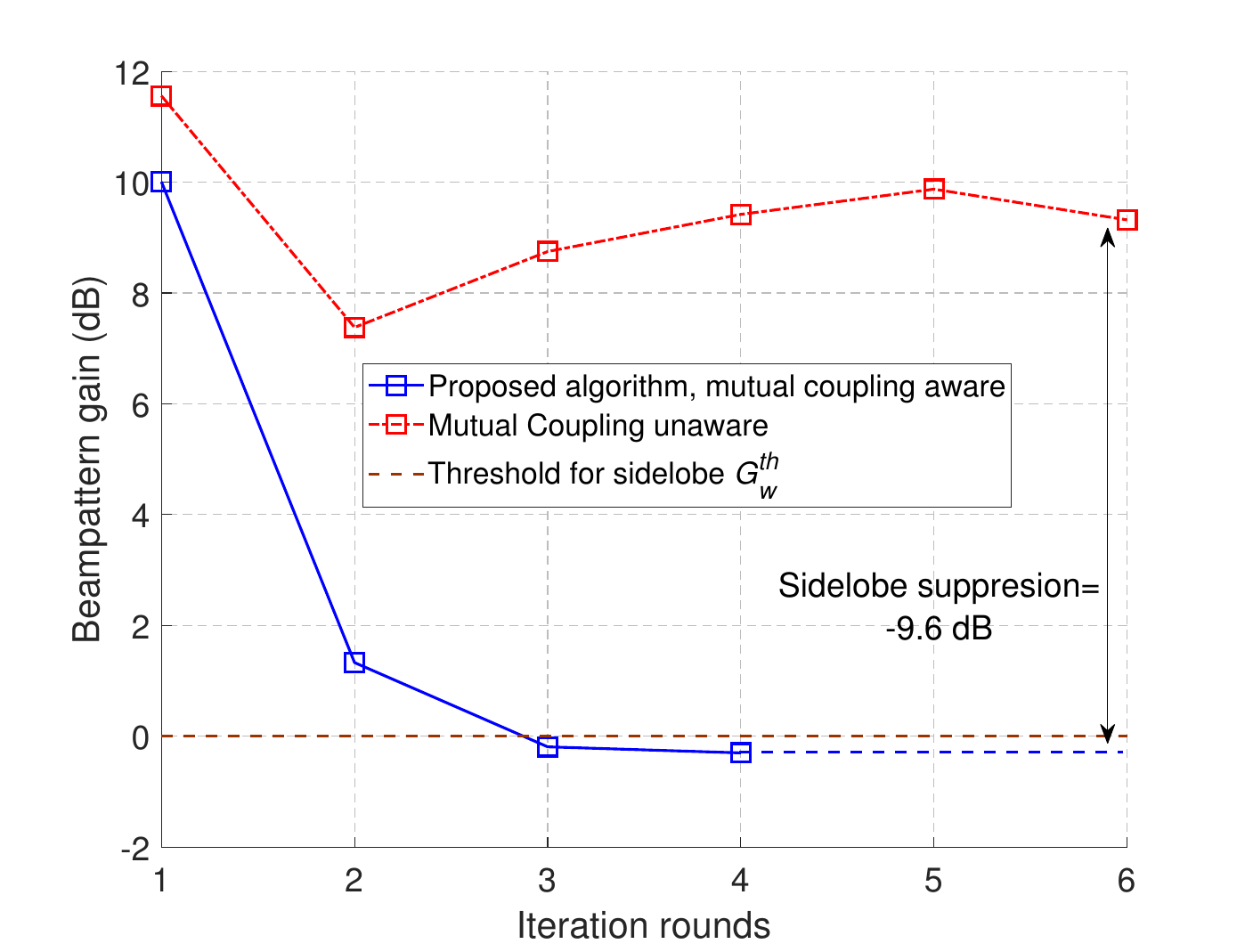}
			%			\vspace{-0.4cm}
			%		\caption{Amplitude vs. frequency}
					\vspace{-0.3cm}
			\label{Beampattern_gain_vs_iteration_rounds_sidelobe_close_clutter}
	\end{minipage}}
	
	\subfigure[Beampattern gain in the direction of sensing targets vs. iteration rounds]{
		\begin{minipage}[b]{0.43\textwidth}
			\centering
						\vspace{-0.2cm}
			\includegraphics[width=.85\textwidth]{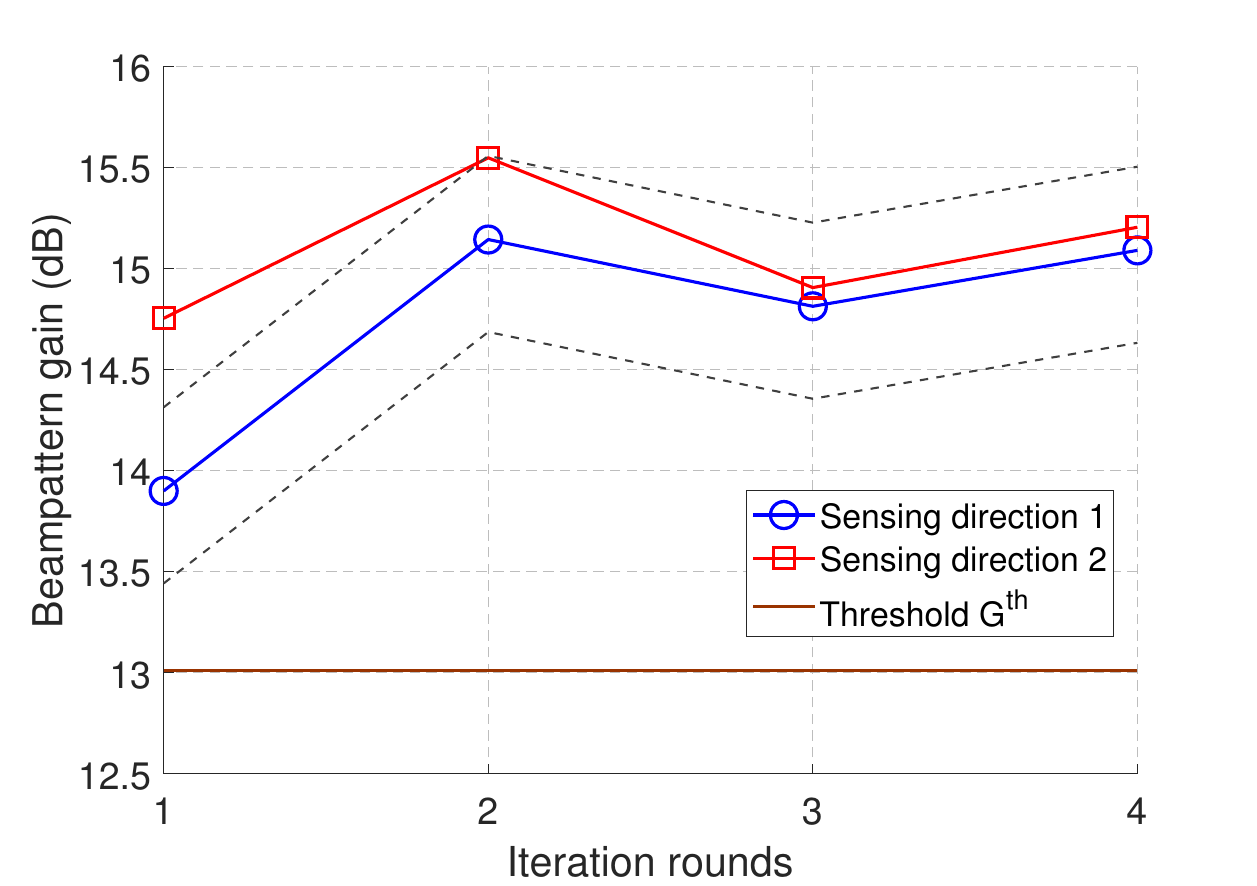}
			%			\vspace{-0.4cm}
			%		\caption{Phase response vs. frequency}
					\vspace{-0.2cm}
			\label{Beampattern_gain_vs_iteration_rounds_close_clutter}
	\end{minipage}}
	%	\vspace{-0.1cm}
	\caption{Beampattern gain versus iteration rounds, with the angle between the clutterer and the closest sensing target equal to $10.6^\circ$. The scenario considers $T=4$ RF chains. The threshold for the beampattern gain in the sensing directions $G^{th}=20$, and that in the clutterer direction $G^{th}_{sb}=1$.}
	\label{fig_beampattern_vs_iteration_rounds_close_clutter}
	\vspace{-0.6cm}
\end{figure}

{Further, in Fig.~\ref{fig_beampattern_vs_iteration_rounds_close_clutter}, we investigate the beampattern gain when the clutterer is closer to sensing targets, i.e., the angle between the clutterer and the closest sensing target reduces to $10.6^\circ$. According to Fig.~\ref{fig_beampattern_vs_iteration_rounds_close_clutter}, the sensing performance constraints are still satisfied, indicating that the proposed algorithm can still ensure sensing performance. Moreover, by comparing Fig.~\ref{fig_beampattern_vs_iteration_rounds_close_clutter} against Fig.~\ref{fig_beampattern_vs_iteration_rounds}, we observe that the sidelobe suppression achieved the proposed mutual coupling aware algorithm is more evident when the clutterer is closer to the sensing target. This highlights the necessity of employing the proposed algorithm in such scenarios.} 
%At low transmit power, inter-user interference is negligible compared to noise, allowing the sum rate to increase with the received signal power. However, as transmit power continues to rise, inter-user interference becomes significant and cannot be ignored. Since both inter-user interference and desired signal power increase with transmit power, the sum rate saturates.

Fig.~\ref{sum_rate_vs_transmit_power} illustrates the impact of maximum allowed sidelobe gain $G_w^{th}$ on the sum rate of the RHS-aided multi-user network. For comparison, two benchmark schemes are also included. In the random algorithm, the holographic pattern at the RHS (i.e., the magnitude of polarizability of the equivalent magnetic dipoles) is randomly assigned, while the digital beamformer is optimized using the same approach as in this paper. In the mutual coupling unaware scheme, mutual coupling is not considered during the joint optimization of the digital beamformer and holographic pattern, which is achieved by setting the coupling matrix $\bm{G} = \bm{0}$. The maximum and minimum polarizability magnitudes are given by $6.02\times10^{-7}$ and $3.21\times10^{-7}$, respectively. Further, the RHS contains $N=36$ elements. {From Fig.~\ref{sum_rate_vs_transmit_power}, we can see that the proposed algorithm outperforms the conventional mutual coupling unaware algorithm in terms of the minimum data rate among all users, highlighting the importance of accounting for mutual coupling effects. Furthermore, the proposed algorithm outperforms the random algorithm, demonstrating the effectiveness of the proposed holographic beamforming method. In addition, we can see that the minimum data rate among users becomes larger when the maximum allowed sidelobe gain $G_w^{th}$ increases, since the increase of the threshold $G_w^{th}$ indicates larger feasibility set. Fig.~\ref{sum_rate_vs_transmit_power} also shows that when the angle between the clutterer and the closest sensing target reduces from $20.2^\circ$ to $10.6^\circ$, the minimum data rate among the users become smaller\footnote{Further analysis shows that when the angle difference between communication users and sensing targets becomes smaller, the communication rates become smaller. This is because for communication users, the beams that serve adjacent communication users or sensing targets can cause interference to them, thus leading to lower communication data rate.}. This is because when the clutterers are closer to the sensing target, the corresponding sensing beam should be more narrowly focused to suppress signal leakage toward the clutterer, thereby meeting the predetermined interference threshold. The resulting tighter beamwidth constraint limits the feasible design space, ultimately leading to degraded optimal objective values, i.e., smaller communication rates.}
\begin{figure}[!t]
	\centering
	\includegraphics[width=0.36\textwidth]{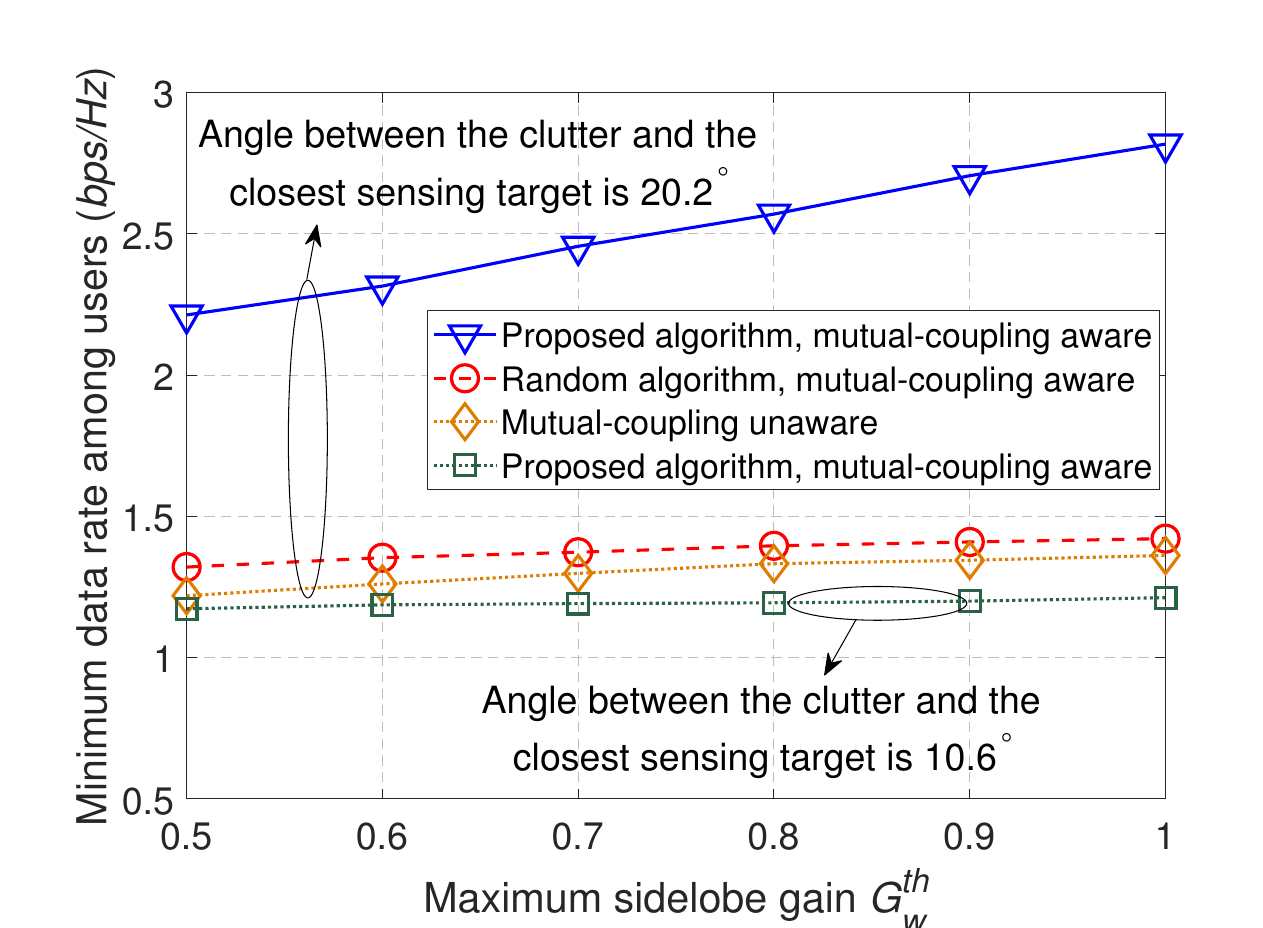}
	%		\vspace{1.5mm}
	\caption{Minimum data rate among users versus maximum allowed sidelobe gain $G_w^{th}$, with the number of RF chains $T=4$, and the threshold for the beampattern gain $G^{th}=20$.}
	\vspace{-5mm}
	\label{sum_rate_vs_transmit_power}
\end{figure}
\begin{figure}[!t]
	\centering
	\includegraphics[width=0.33\textwidth]{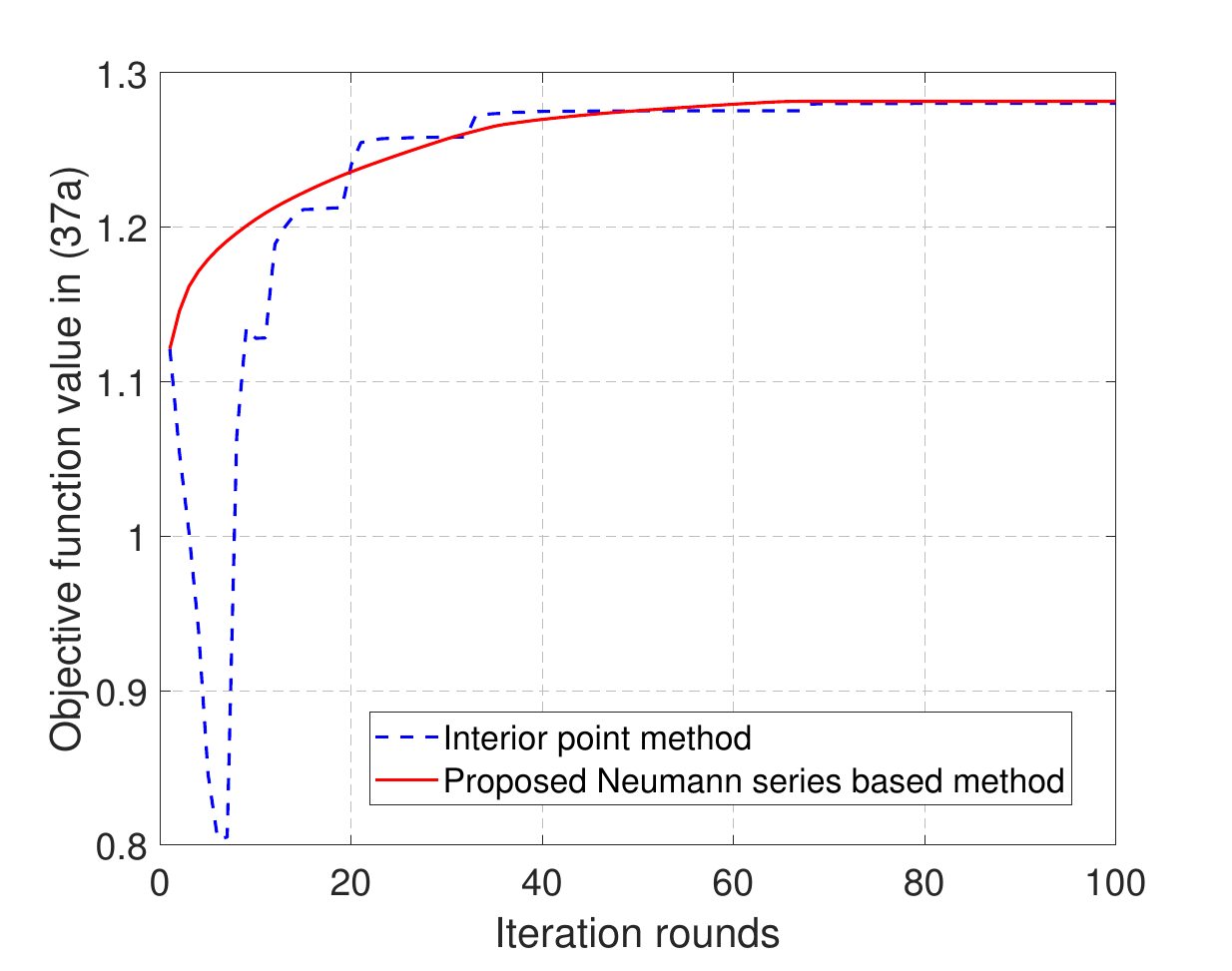}
			\vspace{-.3mm}
	\caption{Objective function value in (37a) vs. iteration rounds. The interior-point based method is implemented by the fmincon function in \textit{Matlab}. We set the number of RF chains $T=4$, the threshold for the beampattern gain in the sensing directions $G^{th}=20$, and that in the clutterer direction $G^{th}_{sb}=1$.}
		\vspace{-4mm}
	\label{fig_convergence_rate}
\end{figure}

{In Fig.~\ref{fig_convergence_rate}, we compare the convergence rate of the proposed method against conventional interior point methods, which are often utilized to deal with constrained optimization problems. Here, the interior-point method is implemented through the fmincon function embedded in \textit{Matlab}. From Fig.~\ref{fig_convergence_rate}, we can see that the convergence rate of the Neumann series based method is comparable to that of existing interior-point based method, and the two algorithms achieve similar objective function values. This indicates that the convergence rate of the proposed method is acceptable. Further, note that by applying the Neumann series approximation, we do not have to perform the inverse matrix operations, and we only need to perform linear matrix operations instead, which is more computationally efficient. This thus motivates us to employ the Neumann series based holographic beamforming method.}

%the clutter moves closer to the sensing targets
%For both the proposed and random algorithms, 
%the data rate increases with transmit power. This behavior arises because inter-user interference is effectively mitigated through digital beamforming, while the received desired signal power scales with transmit power. In contrast, for the mutual coupling unaware algorithm, the data rate increases slowly with transmit power. This is because the algorithm does not consider the mutual coupling effect, and thus is unable to accurately model the influence of RHS-based holographic beamforming on the transmitted signals. Therefore, the digital beamforming algorithm fails to mitigate inter-user interference effectively. Since the inter-user interference also increases with the transmit power, the data rate achieved by the mutual coupling unaware scheme increases more slowly with the transmit power compared with the other two mutual coupling aware schemes. 

%Moreover, from Fig.~\ref{sum_rate_vs_transmit_power}, we can find that 

\begin{figure}[!t]
	\centering
	\includegraphics[width=0.36\textwidth]{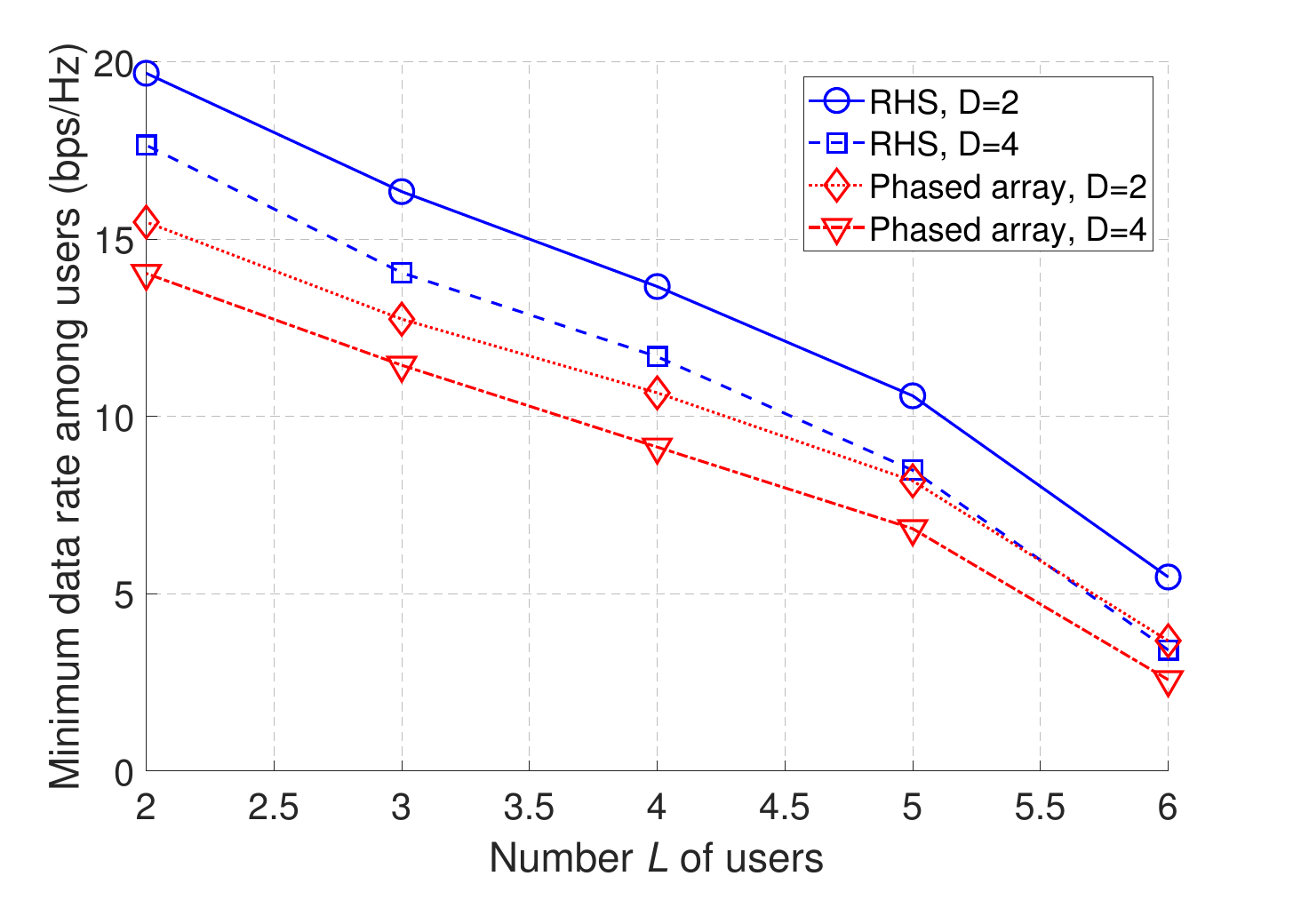}
			\vspace{-1.5mm}
	\caption{Minimum data rate among users versus number of users $L$, under different numbers $D$ of sensing directions. The hardware cost of the RHS is the same as that of the phased array~\cite{Hu_holographic_LEO_2023}. The BS contains $T=6$ RF chains.}
	\vspace{-5mm}
	\label{sum_rate_vs_feed_number}
\end{figure}

Fig.~\ref{sum_rate_vs_feed_number} compares the performance of an ISAC system employing an RHS with that using a conventional phased array given their hardware costs. We assume that the hardware cost of the RHS is the same as that of the phased array. From Fig.~\ref{sum_rate_vs_feed_number}, we observe that despite the presence of the mutual coupling effect, the RHS still achieves a higher sum rate than conventional phased arrays, owing to its larger array aperture. Furthermore, the minimum data rate among the users decrease with the number $L$ of users due to reduced transmit power allocated to each user. Further, we can see that as the number $D$ of sensing directions becomes larger, the data rate also degrades. This is because according to (\ref{cons_fair_sensing}) and (\ref{cons_min_sensing}), a larger $D$ indicates more constraints to the holographic beamforming problem, leading to fewer feasible solutions and thus lower optimal value (i.e., lower data rate). 

\vspace{-.3cm}
\section{Conclusion}
\vspace{-.1cm}
\label{sec_conclusion}
In this paper, we have studied an RHS-enabled holographic ISAC system in clutter environments, explicitly accounting for the mutual coupling among RHS elements using coupled dipole approximations. To suppress sidelobes and enhance ISAC performance, we have formulated a holographic ISAC problem by jointly optimizing the digital beamforming at the BS and the holographic beamforming at the RHS, subject to maximum allowed beamforming gain in the direction of clutterers. To address this non-convex optimization problem, we have decomposed it into two subproblems, and have proposed a mutual coupling aware joint optimization algorithm to solve them alternatively. Based on our analysis and numerical results, the following conclusions can be drawn:
\begin{itemize}
	\item When the mutual coupling effects are not considered in holographic beamforming, the RHS array at ISAC BS can generate notable undesired sidelobes, since the coupled fields within the RHS can disrupt reference wave distributions. 
%	the beam generated by the RHS (using existing holographic beamforming methods) can deviate from intended communication users and sensing directions with severe sidelobes, leading to degraded holographic ISAC system performance.
%	In both single-user and multi-user case, the mutual coupling effect can render the beam generated by the RHS deviate from intended users with severe sidelobes, thus deteriorating system performances.
	\item With the proposed mutual coupling aware beamforming algorithm, the sidelobe of the RHS can be effectively suppressed, thus reducing echo signals from environmental clutterers and improving sensing performance.
	
%	 can be precisely directly towards communication users and sensing directions with suppressed sidelobes. This enables a higher data rate than existing mutual coupling unaware beamforming schemes while ensuring sensing performances.
	
	\item When the mutual coupling effects within RHS arrays are considered, ISAC systems employing an RHS can still outperform conventional phased-array-based ISAC systems given the same hardware cost.
\end{itemize}

%characterized mutual coupling among RHS elements by utilizing coupled dipole approximations, based on which an RHS-assisted downlink communication model is established. To compensate for performance degradation brought by the mutual coupling effect, we formulate a sum-rate maximization problem by jointly optimizing the digital beamforming at the BS, the holographic beamforming at the RHS, and the receive combining at the UE. Since the optimization variables are coupled, it is non-trivial to solve the sum-rate maximization problem. To cope with this issue, we decompose the sum-rate maximization problem into three subproblems, i.e., digital beamforming subproblem, holographic beamforming subproblem, and receive combining subproblem. Then, a mutual coupling aware joint optimization algorithm is proposed to solve these subproblems in an alternating manner. Simulation results demonstrate the effectiveness of the proposed algorithm.

\begin{appendices}
	
%	By stacking the polarizability of different equivalent dipoles as a diagonal matrix $\bm{\Theta}$, (\ref{model_scalar}) can be reorganized as
%	\begin{align}
%		(\bm{\Theta}^{-1}-\bm{G})\bm{m}=\bm{f}^{ref}_l,
%	\end{align}
%	where $\bm{G}$ is a matrix consisting
\vspace{-.2cm}
%\vspace{0.2cm}
	\section{Proof of Remark~\ref{remark_eff}}
	\label{app_eff}
Consider the special case of an ideal RHS without mutual coupling effects, where the transmitted signals of the RHS elements will not propagate to the other RHS elements and interact with them. Therefore, all entries within the coupling matrix $\bm{G}$ are equal to zero, i.e. $\bm{G}=\bm{0}$. 
%	\begin{align}
%		\bm{G}=\bm{0}.
%	\end{align}
Then, the holographic beamformer in (\ref{Holographic_BF}) can be rewritten as $\bm{B}=ke^{j\tau}\bm{\Theta}\widetilde{\bm{F}}^{ref}$, which is found to be consistent with that intended for an ideal RHS without mutual coupling (i.e., (\ref{ideal_Holographic_BF})). This demonstrates the generalizability of the proposed holographic beamformer.

\vspace{-.2cm}

{
\section{RHS radiation pattern in relation to reference wave and geometric channel model}
\label{app_radiation_pattern_analysis}
Here, we show how the reference wave and the geometric channel model are linked to the radiation pattern of the RHS. 

Specifically, the radiation pattern of the RHS corresponding to the $t$-th feed of the RHS, i.e., the radiation pattern when transmit signals are injected only to the $t$-th feed while leaving the other feeds idle, can be given by $P_t(\theta, \phi) = \mathbb{E} \left(z_t(\theta, \phi)z_t^*(\theta, \phi) \right) = \bm{a}^T(\theta, \phi) \bm{b}_t \bm{b}_t^H \bm{a}^*(\theta,\phi)$.
%\begin{equation}
%	\label{radiation_pattern_RHS_v2}
%	P_t(\theta, \phi) = \mathbb{E} \left(z_t(\theta, \phi)z_t^*(\theta, \phi) \right) = \bm{a}^T(\theta, \phi) \bm{b}_t \bm{b}_t^H \bm{a}^*(\theta,\phi),
%\end{equation}
Here, $z_t(\theta, \phi)$ is the far-field signal towards the direction $(\theta,\varphi)$, $\bm{b}_t$ is the $t$-th column of the holographic beamformer $\bm{B}$. By substituting the expression of the holographic beamformer in (4) of the paper into the expression of the radiation pattern, we can derive a closed-form relationship between the radiation pattern and the reference wave, as follows:
\begin{align}
	\label{radiation_pattern_RHS_vs_ref_wave_v2}
	P_t(\theta, \phi) =& |k|^2 \bm{a}^T(\theta, \phi) \left((e^{j\tau}\bm{\Theta})^{-1}-\bm{G}\right)^{-1}\widetilde{\bm{f}}_t^{ref} (\widetilde{\bm{f}}^{ref}_t)^H\notag\\ &\left(\left((e^{j\tau}\bm{\Theta})^{-1}-\bm{G}\right)^{-1}\right)^H\bm{a}^*(\theta,\phi).
\end{align}
Here, $\widetilde{\bm{f}}_{t}^{ref}$ represents the reference wave distribution generated by the $t$-th feed of the RHS.

Then, we show how the radiation pattern influences the geometric model for the equivalent channel from the RF chains of the BS to the UEs. Specifically, assume the multipath channel from the RHS to the $l$-th user consists of $P_l$ paths. Denote the angle-of-departure corresponding to the $p$-th path as $(\theta_p^{l},\varphi_p^{l})$. Further, define $\alpha_p^l$ as the complex coefficient corresponding to the $p$-th path. Then, the geometric model for the equivalent channel $\bm{h}^{(l)}\in\mathbb{C}^{T\times 1}$ from the $T$ RF chains of the BS to the $l$-th UE can be written as
\begin{align}
	\label{geo_channel_v2}
%	\bm{h}^{(l)}&=(\sum_{p=1}^{P_l}\alpha_p^{(l)}\bm{a}^T(\theta_p^{l},\varphi_p^{l})\bm{B})^T\notag\\
%	&=\sum_{p=1}^{P_l}\alpha_p^{(l)}\bm{B}^T\bm{a}(\theta_p^{l},\varphi_p^{l}).
\bm{h}^{(l)}\!=\!(\sum_{p=1}^{P_l}\!\alpha_p^{(l)}\bm{a}^T(\theta_p^{l},\varphi_p^{l})\bm{B})^T\!=\!\sum_{p=1}^{P_l}\!\alpha_p^{(l)}\bm{B}^T\bm{a}(\theta_p^{l},\varphi_p^{l}).
\end{align}
Further, we have
\begin{align}
	\label{radiation_pattern_v2}
	&\bm{B}^T\bm{a}(\theta_p^{l},\varphi_p^{l})=[\bm{a}^T(\theta_p^{l},\varphi_p^{l})\bm{b}_1,\dots,\bm{a}^T(\theta_p^{l},\varphi_p^{l})\bm{b}_T]^T,\notag\\
	&=[\sqrt{P_1(\theta_p^{l},\varphi_p^{l})}e^{j\phi_{1,p}^{(l)}},\dots,\sqrt{P_T(\theta_p^{l},\varphi_p^{l})}e^{j\phi_{T,p}^{(l)}}]^T.
\end{align}
where $\bm{b}_t$ is the $t$-th column of the holographic beamformer $\bm{B}$, $P_t(\theta_p^{l},\varphi_p^{l})$ is the radiation pattern in the direction $(\theta_p^{l},\varphi_p^{l})$ corresponding to the $t$-th feed of the RHS, and $\phi_{t,p}^{(l)}$ is the phase of term $\bm{a}^T(\theta_p^{l},\varphi_p^{l})\bm{b}_T$. By substituting (\ref{radiation_pattern_v2}) into (\ref{geo_channel_v2}), we have a closed-form relationship between the geometric channel and the radiation of the RHS, i.e.,
%\begin{align}
%	\label{geo_channel_link_v2}
%	&\bm{h}^{(l)}=\notag\\
%	&[\sum_{p=1}^{P_l}\alpha_p^{(l)}\sqrt{P_1(\theta_p^{l},\varphi_p^{l})}e^{j\phi_{1,p}^{(l)}},\dots,\sum_{p=1}^{P_l}\alpha_p^{(l)}\sqrt{P_T(\theta_p^{l},\varphi_p^{l})}e^{j\phi_{T,p}^{(l)}}]^T.
%\end{align}
\begin{align}
	\label{geo_channel_link_v2}
	\bm{h}^{(l)}=[&\sum_{p=1}^{P_l}\alpha_p^{(l)}\sqrt{P_1(\theta_p^{l},\varphi_p^{l})}e^{j\phi_{1,p}^{(l)}},\notag\\
	&\dots,\sum_{p=1}^{P_l}\alpha_p^{(l)}\sqrt{P_T(\theta_p^{l},\varphi_p^{l})}e^{j\phi_{T,p}^{(l)}}]^T.
\end{align}
}
\vspace{-7mm}
\section{Proof of Lemma~\ref{lemma_approx_HB}}
\label{app_approx_HB}
According to (\ref{B_k}), we can rewrite $\bm{B}_k$ as $\bm{B}_k=k(e^{-j\tau}\bm{\Theta}_{k-1}^{-1}-\bm{G}-e^{-j\tau}\delta_k\widetilde{\bm{\Theta}})^{-1}\widetilde{\bm{F}}^{ref}$.
%\begin{align}
%	\label{B_k_2}
%	\bm{B}_k=k(e^{-j\tau}\bm{\Theta}_{k-1}^{-1}-\bm{G}-e^{-j\tau}\delta_k\widetilde{\bm{\Theta}})^{-1}\widetilde{\bm{F}}^{ref}.
%\end{align}
Recall that $\bm{S}_k=e^{-j\tau}((e^{j\tau}\bm{\Theta}_{k-1})^{-1}-\bm{G})^{-1}$. Therefore, we have 
\begin{align}
%	\label{B_k_3}
	\bm{B}_k&=k\big((e^{-j\tau}\bm{\Theta}_{k-1}^{-1}-\bm{G})(\bm{I}-\delta_k\bm{S}_k\widetilde{\bm{\Theta}})\big)^{-1}\widetilde{\bm{F}}^{ref},\notag\\
	\label{B_k_4}
	&=k(\bm{I}-\delta_k\bm{S}_k\widetilde{\bm{\Theta}})^{-1}(e^{-j\tau}\bm{\Theta}_{k-1}^{-1}-\bm{G})^{-1}\widetilde{\bm{F}}^{ref}.
\end{align} 
Then, based on Neumann series, the expression $(\bm{I}-\delta_k\bm{S}_k\widetilde{\bm{\Theta}})^{-1}$ included in (\ref{B_k_4}) can be rewritten as
\begin{align}
	\label{neumman_2}
	(\bm{I}-\delta_k\bm{S}_k\widetilde{\bm{\Theta}})^{-1}=\bm{I}+\sum_{i=1}^{\infty}(\delta_k\bm{S}_k\widetilde{\bm{\Theta}})^i.
\end{align}
Note that $\delta_k$ takes a small value. Therefore, we can use the first two terms in the right-hand-side of (\ref{neumman_2}) to approximate the expression in its left-hand-side, i.e.,
\begin{align}
	\label{neumman_3}
	(\bm{I}-\delta_k\bm{S}_k\widetilde{\bm{\Theta}})^{-1}\approx\bm{I}+\delta_k\bm{S}_k\widetilde{\bm{\Theta}}.
\end{align}
By substituting (\ref{neumman_3}) into (\ref{B_k_4}), we can derive the approximation for the holographic beamformer $\bm{B}_k$ given in Lemma~\ref{lemma_approx_HB}.

Then, we show when the approximation in (\ref{neumman_3}) is valid. Define $\bm{\Gamma}$ and $\hat{\bm{\Gamma}}$ as the left-hand-side and right-hand-side of (\ref{neumman_3}), respectively. The approximation in (\ref{neumman_3}) is valid when
\begin{align}
	\label{condition_1}
	\|\bm{\Gamma}-\hat{\bm{\Gamma}}\|\ll 1.
\end{align} 
Note that $\|\bm{\Gamma}-\hat{\bm{\Gamma}}\|\le \frac{\|\delta_k\bm{S}_k\widetilde{\bm{\Theta}}\|^2}{1-\|\delta_k\bm{S}_k\widetilde{\bm{\Theta}}\|}\le \|\delta_k\bm{S}_k\widetilde{\bm{\Theta}}\|^2$.
%\begin{align}
%	\|\bm{\Gamma}-\hat{\bm{\Gamma}}\|\le \frac{\|\delta_k\bm{S}_k\widetilde{\bm{\Theta}}\|^2}{1-\|\delta_k\bm{S}_k\widetilde{\bm{\Theta}}\|}\le \|\delta_k\bm{S}_k\widetilde{\bm{\Theta}}\|^2.
%\end{align}
Therefore, to guarantee that (\ref{condition_1}) holds, we require that
\begin{align}
	\label{condition_2}
	\|\delta_k\bm{S}_k\widetilde{\bm{\Theta}}\|\ll 1.
\end{align}
Further, note that $\|\delta_k\bm{S}_k\widetilde{\bm{\Theta}}\|\le \|\bm{S}_k\|\|\delta_k\widetilde{\bm{\Theta}}\|\le \delta_k \|\bm{S}_k\|$.
%\begin{align}
%	\|\delta_k\bm{S}_k\widetilde{\bm{\Theta}}\|\le \|\bm{S}_k\|\|\delta_k\widetilde{\bm{\Theta}}\|\le \delta_k \|\bm{S}_k\|.
%\end{align}
Therefore, to satisfy (\ref{condition_2}), we require that $\delta_k\ll \|\bm{S}_k\|^{-1}$.
%\begin{align}
%	\delta_k\ll \|\bm{S}_k\|^{-1}.
%\end{align}

%\vspace{-.3cm}	
%	\section{Proof of Proposition~\ref{prop_recover_rank_one_solution}}
%	\label{app_recover_rank_one_solution}
%	...
	\vspace{-.3cm}
	\section{Definitions of Parameters Mentioned in Theorem~\ref{theorem_equivalent_HBF}}
	\label{appendix_equivalent_HBF}
	In the following, we present the definitions of constant matrices and scalars mentioned in (\ref{opt_v4_problem_HBF}). 
%	Then, we demonstrate how to transform problem (\ref{opt_v4_problem_HBF}) into its equivalent form in (\ref{opt_v5_problem_HBF}). Specifically, we have
	\begin{align}
		\bm{C}_n&=\begin{bmatrix}
			\bm{0}_{N\times N} & \frac{1}{2}\bm{1}_n\\
			\frac{1}{2}\bm{1}_n^T & 0
		\end{bmatrix},\hfill\\
		\hat{\bm{U}}_{k,l}&=\begin{bmatrix}
			|\rho_l|^2\bm{U}_{k,l}-|\rho_l|^2\bm{U}_{k} & 2\Re(\bm{f}_{k,l})\\
			2\Re(\bm{f}_{k,l})^T & 0
		\end{bmatrix},\hfill\\
		\widetilde{\bm{U}}_{k,w}&=\begin{bmatrix}
			\bm{U}_{k,w} & \Re(\bm{c}_{k,w})\\
			\Re(\bm{c}_{k,w})^T & 0
		\end{bmatrix},\hfill\\
		\widetilde{\bm{U}}_{k,d}&=\begin{bmatrix}
			\bm{U}_{k,d} & \Re(\bm{c}_{k,d})\\
			\Re(\bm{c}_{k,d})^T & 0
		\end{bmatrix},\hfill\\
		\widetilde{\bm{U}}_{k}&=\begin{bmatrix}
			\bm{U}_{k} & \Re(\bm{c}_{k})\\
			\Re(\bm{c}_{k})^T & 0
		\end{bmatrix},\\
		\widetilde{\bm{U}}_{k,l}&=\begin{bmatrix}
			\bm{U}_{k,l} & \Re(\bm{c}_{k,l})\\
			\Re(\bm{c}_{k,l})^T & 0
		\end{bmatrix},\\
		\bm{U}_{k,l}&=(\bm{S}_k\widetilde{\bm{F}}_{ref}\bm{Q}_l\widetilde{\bm{F}}_{ref}^H\bm{S}_k^H)\circ(\delta_k^2|k|^2\bm{S}_k^H\bm{h}_l^*\bm{h}_l^T\bm{S}_k)^T,\\
		\bm{U}_{k}&=(\bm{S}_k\widetilde{\bm{F}}_{ref}\bm{Q}\widetilde{\bm{F}}_{ref}^H\bm{S}_k^H)\circ(\delta_k^2|k|^2\bm{S}_k^H\bm{h}_l^*\bm{h}_l^T\bm{S}_k)^T,\\
		\bm{U}_{k,w}&=(\bm{S}_k\widetilde{\bm{F}}_{ref}\bm{Q}\widetilde{\bm{F}}_{ref}^H\bm{S}_k^H)\notag\\
		&\quad\circ(\delta_k^2|k|^2\bm{S}_k^H\bm{a}^*(\theta_w,\phi_w)\bm{a}^T(\theta_w,\phi_w)\bm{S}_k)^T,\\
		\bm{U}_{k,d}&=(\bm{S}_k\widetilde{\bm{F}}_{ref}\bm{Q}\widetilde{\bm{F}}_{ref}^H\bm{S}_k^H)\notag\\
		&\quad\circ(\delta_k^2|k|^2\bm{S}_k^H\bm{a}^*(\theta_d,\phi_d)\bm{a}^T(\theta_d,\phi_d)\bm{S}_k)^T,\\
		\bm{d}_{k,l}&=(\rho_lke^{j\tau}\delta_k)^*\bm{v}_{l,c}^H\widetilde{\bm{F}}_{ref}^H\bm{S}_k^H\circ (\bm{h}_l^H\bm{S}_k^H),\\
		\bm{c}_{k,l}&=(ke^{j\tau}\delta_k\bm{h}_l^T\bm{S}_k)^T\circ(k^*e^{-j\tau}\bm{S}_k\widetilde{\bm{F}}_{ref}\bm{Q}_l\widetilde{\bm{F}}_{ref}^H\bm{S}_k^H\bm{h}_l^*)\\
		\bm{c}_{k}&=(ke^{j\tau}\delta_k\bm{h}_l^T\bm{S}_k)^T\circ(k^*e^{-j\tau}\bm{S}_k\widetilde{\bm{F}}_{ref}\bm{Q}\widetilde{\bm{F}}_{ref}^H\bm{S}_k^H\bm{h}_l^*)\\
		e_{k,l}&=(\rho_lke^{j\tau})^*\bm{v}_{l,c}^H\widetilde{\bm{F}}_{ref}^H\bm{S}_k^H\bm{h}_l^*,\\
		b_k&=|k|^2\bm{h}_l^T\bm{S}_k\widetilde{\bm{F}}_{ref}\bm{Q}\widetilde{\bm{F}}_{ref}^H\bm{S}_k^H\bm{h}_l^*,\\
		b_{k,l}&=|k|^2\bm{h}_l^T\bm{S}_k\widetilde{\bm{F}}_{ref}\bm{Q}_l\widetilde{\bm{F}}_{ref}^H\bm{S}_k^H\bm{h}_l^*,\\
		b_{k,w}&=|k|^2\bm{a}^T(\theta_w,\phi_w)\bm{S}_k\widetilde{\bm{F}}_{ref}\bm{Q}\widetilde{\bm{F}}_{ref}^H\bm{S}_k^H\bm{a}^*(\theta_w,\phi_w),\\
		b_{k,d}&=|k|^2\bm{a}^T(\theta_d,\phi_d)\bm{S}_k\widetilde{\bm{F}}_{ref}\bm{Q}\widetilde{\bm{F}}_{ref}^H\bm{S}_k^H\bm{a}^*(\theta_d,\phi_d),\\
		\bm{f}_{k,l}&=\bm{d}_{k,l}-|\rho_l|^2\bm{c}_k+|\rho_l|^2\bm{c}_{k,l}.		
	\end{align}
	Here, $\circ$ represents the element-wise multiplication of two matrices, and $\bm{1}_n\in\mathbb{R}^{N\times 1}$ is a vector consisting of $1$ in the $n$-th element and $0$ for all the other elements.
	
%	Then, we show how to equivalently transform problem (\ref{opt_v4_problem_HBF}) into (\ref{opt_v5_problem_HBF}). 
\end{appendices}

\end{document}